%% file: paper.tex
\documentclass[sigconf]{acmart}
\pdfoutput=1

\newcommand{\ARXIV} % Ben: comment out to make paper shorter (KDD)

\usepackage{booktabs} % For formal tables

\setcopyright{none}  % Ben: perhaps do this for now to save space?
\ifdefined\ARXIV
\acmConference[ ]{ }{ }{ }
\acmPrice{ }
\else
\acmConference[KDD'17]{ACM SIGKDD Conference on Knowledge Discovery and Data Mining}{August 2017}{Halifax, Canada} 
\acmDOI{10.475/123_4}
\acmISBN{123-4567-24-567/08/06}
\acmPrice{15.00}
\fi

\acmYear{2017}
\copyrightyear{2017}

\input{subfiles/user-defined}

\begin{document}
	\sloppy
	\allowdisplaybreaks % to allow page breaks inside equations

\title{End-to-End Differentially-Private \\ Parameter Tuning in Spatial Histograms}
\renewcommand{\shorttitle}{End-to-End Differentially-Private Parameter Tuning in Spatial Histograms} % Ben: added this so that the running title can leave out the newline
%\titlenote{Produces the permission block, and 
%	copyright information}
%\subtitle{Extended Abstract}

%%%%%%%%%%%%%%%%%%%%%%%%%%%%%%%%%%%%%%%%%%%%%%%%%%%%%%
%%% AUTHORS NAME
%%%%%%%%%%%%%%%%%%%%%%%%%%%%%%%%%%%%%%%%%%%%%%%%%%%%%%

\author{Maryam Fanaeepour}
\orcid{0000-0003-2654-982X}
\affiliation{%
\institution{School of Computing and Information System\\
	The University of Melbourne, Australia}
%\country{Australia}
%\state{Victoria} 
%\postcode{3010}
}
%\affiliation{%
%	\institution{Data61, CSIRO, Australia}
%}
\email{mfanaeepour@unimelb.edu.au}

\author{Benjamin I. P. Rubinstein}
\orcid{0000-0002-2947-6980}
\affiliation{%
	\institution{School of Computing and Information System\\
		The University of Melbourne, Australia}
%	\country{Australia}
%	\state{Victoria} 
%	\postcode{3010}
}
\email{brubinstein@unimelb.edu.au}

% The default list of authors is too long for headers
%\renewcommand{\shortauthors}{B. Trovato et al.}

%%%%%%%%%%%%%%%%%%%%%%%%%%%%%%%%%%%%%%%%%%%%%%%%%%%%%%
%%% ABSTRACT
%%%%%%%%%%%%%%%%%%%%%%%%%%%%%%%%%%%%%%%%%%%%%%%%%%%%%%

\input{subfiles/abstract}

%%%%%%%%%%%%%%%%%%%%%%%%%%%%%%%%%%%%%%%%%%%%%%%%%%%%%%
%%% CCS Concepts and Keywords
%%%%%%%%%%%%%%%%%%%%%%%%%%%%%%%%%%%%%%%%%%%%%%%%%%%%%%

%
% The code below should be generated by the tool at
% http://dl.acm.org/ccs.cfm
% Please copy and paste the code instead of the example below. 
%
\begin{CCSXML}
	<ccs2012>
	<concept>
	<concept_id>10002951.10003227.10003236.10003101</concept_id>
	<concept_desc>Information systems~Location based services</concept_desc>
	<concept_significance>500</concept_significance>
	</concept>
	<concept>
	<concept_id>10002978.10002991.10002995</concept_id>
	<concept_desc>Security and privacy~Privacy-preserving protocols</concept_desc>
	<concept_significance>300</concept_significance>
	</concept>
	<concept>
	<concept_id>10002978.10003018.10003019</concept_id>
	<concept_desc>Security and privacy~Data anonymization and sanitization</concept_desc>
	<concept_significance>300</concept_significance>
	</concept>
	<concept>
	<concept_id>10003752.10010070.10010111.10011735</concept_id>
	<concept_desc>Theory of computation~Theory of database privacy and security</concept_desc>
	<concept_significance>300</concept_significance>
	</concept>
	</ccs2012>
\end{CCSXML}

\ccsdesc[500]{Information systems~Location based services}
\ccsdesc[300]{Security and privacy~Privacy-preserving protocols}
\ccsdesc[300]{Security and privacy~Data anonymization and sanitization}
\ccsdesc[300]{Theory of computation~Theory of database privacy and security}

% We no longer use \terms command
%\terms{Theory}

\keywords{Differential Privacy, %Grid Data Structure,
Histograms, Location Privacy}

\maketitle

%%%%%%%%%%%%%%%%%%%%%%%%%%%%%%%%%%%%%%%%%%%%%%%%%%%%%%
%%% Body
%%%%%%%%%%%%%%%%%%%%%%%%%%%%%%%%%%%%%%%%%%%%%%%%%%%%%%
\input{subfiles/intro}

\input{subfiles/relatedwork}

\input{subfiles/preliminary}

\input{subfiles/problem}

\input{subfiles/approach}

\input{subfiles/analysis}

\input{subfiles/experiment}

\input{subfiles/conclusion}

\section*{Acknowledgements}

This work was supported in part by the Australian Research Council through grant DE160100584, and Data61/CSIRO through a NICTA PhD Scholarship.

\bibliographystyle{ACM-Reference-Format}
\bibliography{refs} % references in refs.bib

%%%%%%%%%%%%%%%%%%%%%%%%%%%%%%%%%%%%%%%%%%%%%%%%%%%%%%
%%% Appendix
%%%%%%%%%%%%%%%%%%%%%%%%%%%%%%%%%%%%%%%%%%%%%%%%%%%%%%
\ifdefined\ARXIV
\appendix
\input{subfiles/appendix}
\fi

\end{document}

%% file: subfiles/user-defined.tex
\usepackage{graphicx}
\usepackage[ruled,vlined]{algorithm2e}
\usepackage[caption=false]{subfig}
\usepackage{float}
\usepackage{xspace}
\usepackage{paralist}
\usepackage{amssymb}
\usepackage{amsmath}
\usepackage{amsfonts}
\usepackage{mathtools}
\usepackage{bm} % For bold greek letters
\usepackage{url}
\usepackage{enumerate}
\DeclarePairedDelimiter\ceil{\lceil}{\rceil}

\DeclarePairedDelimiter\abs{\lvert}{\rvert}
\DeclarePairedDelimiter\norm\lVert\rVert

\DeclareMathAlphabet{\mathcal}{OMS}{cmsy}{m}{n}
\newcommand{\eg}[0]{\emph{e.g.,}\xspace}
\newcommand{\ie}[0]{\emph{i.e.,}\xspace}
\newcommand{\cf}[0]{\emph{cf.}\xspace}
\newcommand{\etal}[0]{\emph{et~al.}\xspace}
\newtheorem{rem}[theorem]{Remark}
\newtheorem{prob}[theorem]{Problem}
\newtheorem{property}[theorem]{Property}

\newcommand{\term}[1]{\emph{#1}}
\newcommand{\myparagraph}[1]{\noindent \textbf{#1}\xspace}

\newcommand{\QRset}[0]{\ensuremath{\mathcal{Q}}\xspace} % Query Rectangles set

\newcommand{\Cset}[0]{\ensuremath{\mathcal{C}}\xspace} % Cell counts set inside the QR

\renewcommand{\c}[1]{\ensuremath{c_{#1}}\xspace} % Cell count component

\newcommand{\Dset}[0]{\ensuremath{D}\xspace} % Original dataset

\newcommand{\DNset}[0]{\ensuremath{D'}\xspace}  % Neighbor dataset to D

\renewcommand{\d}[1]{\ensuremath{d_{#1}}\xspace} % true number of data in cell i overlapping QR

\newcommand{\Gset}[0]{\ensuremath{\mathcal{G}}\xspace} % set of grid cell sizes

\newcommand{\g}[1]{\ensuremath{g_{#1}\xspace}} % grid cell size

\newcommand{\Hb}[0]{\ensuremath{\mathbf{H}}\xspace} % Histogram vector

\newcommand{\Hn}[0]{\ensuremath{\mathbf{H'}}\xspace} % Noisy Histogram vector

\newcommand{\gstar}[0]{\ensuremath{g^\star}\xspace} % selected cell size after DP tuning

\newcommand{\QR}[0]{\ensuremath{QR}\xspace} % Query Region

\newcommand{\qr}[1]{\ensuremath{Q_{#1}}\xspace} % Query element of the query regions set

\renewcommand{\a}[1]{\ensuremath{\alpha_{#1}}\xspace} % percentage overlapping between QR and cell i

\newcommand{\Alpha}[0]{\ensuremath{\bm{\alpha}}\xspace} % Vector of overlapping areas between QR and cell i

\newcommand{\Y}[1]{\ensuremath{Y_{#1}}\xspace} % noise added to cell
\newcommand{\Yb}[0]{\ensuremath{\mathbf{Y}}\xspace} % Vector of noise components

\newcommand{\GS}[1]{\ensuremath{\Delta_{#1}}\xspace} % Response dependant sensitivity

\newcommand{\M}[0]{\ensuremath{\mathcal{M}}\xspace}

\newcommand{\Exp}[1]{\ensuremath{\mathbb{E}\left[#1\right]}}
\newcommand{\Prob}[1]{\ensuremath{\mathrm{Pr}(#1)}}

\newcommand{\eat}[1]{}

\newcommand{\reals}{\ensuremath{\mathbb{R}}\xspace}

\newcommand{\EE}[0]{\textsc{E2EPriv}\xspace}
\newcommand{\best}[0]{\textsc{BestNonPriv}\xspace}
\newcommand{\heuristic}[0]{\textsc{Heuristic}\xspace}
\newcommand{\leaky}[0]{\textsc{Leaky}\xspace}

%% file: subfiles/abstract.tex
\begin{abstract}
Differentially-private histograms have emerged as a key
tool for location privacy. While past mechanisms have 
included theoretical \& experimental analysis, it has
recently been observed that much of the existing literature does
not fully provide differential privacy. The missing component,
private parameter tuning, is necessary for rigorous
evaluation of these mechanisms. Instead works frequently tune
on training data to optimise parameters without consideration
of privacy; in other cases selection is performed arbitrarily
and independent of data, degrading utility. We address this
open problem by deriving a principled tuning mechanism that
privately optimises data-dependent error bounds. Theoretical
results establish privacy and utility while extensive 
experimentation demonstrates that we can practically achieve
true end-to-end privacy.
\end{abstract}

%% file: subfiles/intro.tex
\section{Introduction}\label{sec:intro}

Location data is used widely, from ride-sharing apps in
consumer mobile to traffic management in urban planning.
But the utility of location analytics must be balanced with
concerns over user privacy. A leading framework for strong privacy
guarantees suitable to the setting, is differential
privacy~\cite{Cynthia2006,privacybook14}. Many authors have studied the
release of spatial data structures to untrusted third parties,
for accurate response to range queries under differential
privacy~\cite{Cormode12PSD,ChenFDS12,QardajiYL13,QardajiHeirarchy2013,HeCMPS15,Bodies2016}.
However a recent large-scale analysis~\cite{Hay16} has discovered
that reported evaluations in previous work have parameter-tuned
non-privately, undermining the validity of much prior work.
In this paper, we develop private tuning of spatial histograms
through optimising privatised data-dependent error bounds,
addressing the gap on end-to-end privacy (\cf Figure~\ref{fig:E2E_process}).

%Figure~\ref{fig:E2E_process} demonstrates an end-to-end private tuning to release a sanitized spatial data structure in a non-interactive settings. Input data of an original spatial database can be point location, trajectories, or users' regions, which get sanitized and released as a spatial histogram.
\begin{figure}[t!]
	\centering
	\includegraphics[width = .9\columnwidth]{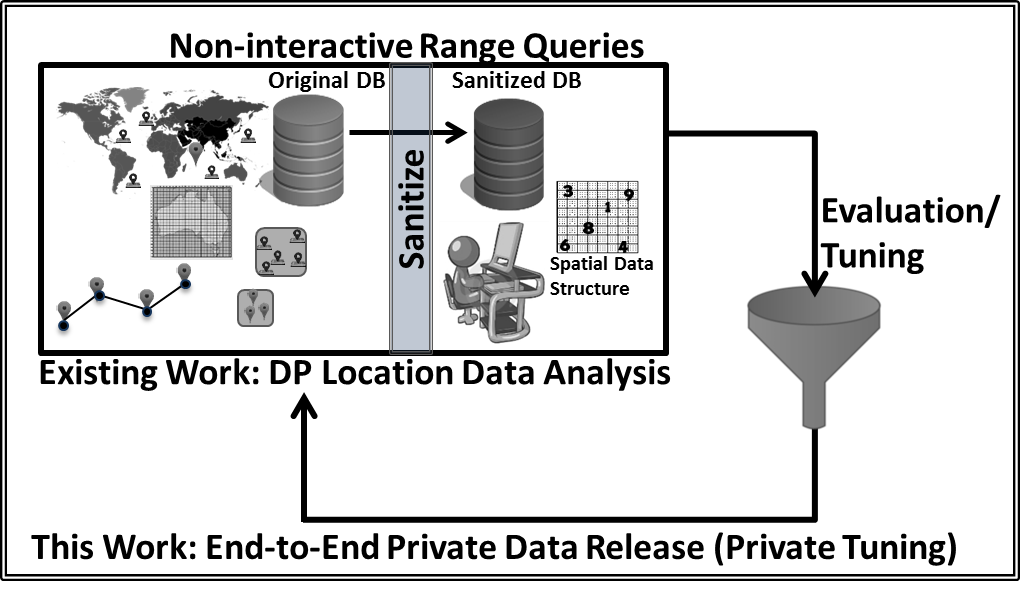}
	\caption{End-to-end private spatial data structures.}
	\label{fig:E2E_process}
\end{figure}

Aggregation has been used extensively 
for efficient range query responses, and as a strategy for
qualitative privacy~\cite{Chawla2005}. Differential privacy 
complements such approaches by addressing attacker background
knowledge; to date, various spatial data structures have been
adopted for private spatial data mining~\cite{Cormode12PSD,ChenFDS12,QardajiYL13,QardajiHeirarchy2013,HeCMPS15,Bodies2016}.

As research has established that utility is highly parameter
dependant~\cite{QardajiYL13,QardajiHeirarchy2013},
parameters must be tuned on data, and therefore privately.
Unfortunately, as documented recently~\cite{Hay16}, many past
works on histogram release establish differential privacy of 
mechanisms while ignoring privacy during tuning. The
\textsc{DPBench} framework~\cite{Hay16}
presented, articulates as an open problem the need for
end-to-end privacy for truly privacy-preserving mechanisms
and fair, rigorous evaluations.

In this paper, we address this problem by optimising privatised
data-dependent error bounds that quantify the effect of data
structure parameters. Our focus is releasing histograms, as these
are the most widely used and effective spatial data
structures~\cite{QardajiYL13}.

Our mechanism consists of runs over two phases: 1) among all values
for the parameter, one is selected privately that is close in
utility to an optimum with high probability; 2) the data
structure is constructed \& released privately using the selected
parameter. The main challenge is bounding utility of phase two 
with respect to phase one's parameter selection, and doing so privately. 
We consider range query relative error~\cite{GarofalakisG02,VitterW99}
as our objective when choosing histogram grid size. Our bounds on
this error decompose into two errors, through a principled analysis:
\term{aggregation} error due to the (common) use of the
uniformity assumption for aggregated counts when data is
non-uniformly distributed; and \term{perturbation} error due 
to count perturbation for phase two differential privacy.

\myparagraph{Contributions.}
Our main contributions include
\begin{compactitem}
	\item For the first time, a solution to end-to-end differentially-private parameter tuning for spatial data structure release;
	\item A two-phase mechanism for private parameter tuning and data structure construction;
	\item Guarantees on differential privacy and utility;
	\item Extensive experimental confirmation that our mechanism is the new state-of-art for private accurate histograms.
\end{compactitem}

%% file: subfiles/relatedwork.tex
\section{Related Work}\label{sec:related_work}
Numerous proposals have sought to address the challenge of
private location-based
services~\cite{Ghinita2013}. Aggregation has widely been used
as a qualitative privacy approach, by reporting aggregate numbers
of objects per partition cell in response to range
queries~\cite{Braz2007, Tao2004,Lopez2005,LeonardiORRSAA14,CASE2015}.
%~\cite{Braz2007, Tao2002,Lopez2005,LeonardiORRSAA14,Papadias2001, CASE2015}.
Differential privacy~\cite{Cynthia2006,privacybook14} has also been
adopted as a semantic definition for privacy when releasing
structures to untrusted third parties. To achieve high
utility, different variants of data structures have been
explored~\cite{QardajiYL13,Cormode12PSD,ChenFDS12,HeCMPS15,Bodies2016}, such as spatial grid histograms,
quad-trees, kd-trees for point locations, for
trajectories, as well as user
regions.

For each data structure, selection of parameters such as
grid size or levels of hierarchies, is known to be of the utmost
importance in affecting utility~\cite{QardajiYL13,QardajiHeirarchy2013,PrivTree2016}.
The authors in~\cite{QardajiYL13} propose
Equation~\eqref{eq:heuristic} as a guideline for selecting grid
size when releasing differentially-private grid-partitioned
synopses: 
\begin{eqnarray}
	m &=& \sqrt{\frac{N\epsilon}{c}}\enspace,\label{eq:heuristic}
\end{eqnarray}
where $m$ is the selected grid size per direction, $N$ is the
number of data points, $\epsilon$ is the total privacy budget
and $c$ is a constant depending on the dataset. Their stated
motivation is to balance perturbation (noise) error and
aggregate (non-uniformity) error, and while they analyse each
error component, their combination is performed without
rigorous justification. Moreover, the authors tune $c$ on their sensitive
experimental datasets, simultaneously undermining: $c$'s definition
as a constant, potentially leaking privacy, and overfitting their
structures to test data. We refer to this grid selection approach
as \heuristic in experiments (\cf Section~\ref{sec:experiment}).

It has been noted that once a parameter is already 
tuned non-privately on past sensitive data, that parameter can be
used safely on future unrelated
datasets~\cite{Hay16}. However, such fixed schemes still eschew optimisation by
data-dependence. Not all datasets exhibit the same levels of 
uniformity, point distribution or domain, as discussed in
Section~\ref{sec:heuristic}. Such approaches like \heuristic
obfuscate the non-privacy of tuning $m$ by secretly tuning $c$
(or some other constant in the fixed rule). Any tuning must be
privacy preserving.

The key challenge for this line of research, is that parameter
selection must be data dependent but still preserve privacy. In
machine learning, private hyper parameter tuning has been
explored~\cite{Chaudhuri11,Chaudhuri13} using cross validation.
However, cross validation leverages split test \& train data,
as it aims to mitigate \emph{future} generalisation error. Here we wish
to make use of all data in all stages and are ultimately concerned
with range queries against this same dataset. The two domains are
related but pose fundamentally distinct challenges.

In~\cite{DAWA14}, a private parameter selection mechanism, for
1D data, is developed using dynamic programming. While it is
speculated that the approach extends to 2D data via reducing
2D structures to 1D with space filling curves, such curves
do not preserve spatial locality in general. As a result it is 
relatively easy to construct counter examples to such extensions.

%In~\cite{DAWA14}, authors discuss a parameter tuning technique for 1D data structure and argue that can be applied for 2D as well, however, in practice their proposed way of applying is not effective. More specifically, they discuss the use of space filling curve (SFC), where any indexed 2D spatial data will be translated to 1D by adopting SFC. However, the main long lasting issue with SFC will remain which is not addressing and preserving the locality of the points and it will significantly affect the utility. In this work we address this gap in the literature and focus on 2D data structure.

A principled evaluation for differentially-private algorithms is
reported recently in~\cite{Hay16}. The \textsc{DPBench} framework
asserts that end-to-end privacy is quite necessary, and highlights parameter tuning as
a key open problem for many existing mechanisms. We are motivated
by their call, and address the problem with our end-to-end
private approach for tuning and histogram construction.

%% file: subfiles/preliminary.tex
\section{Preliminaries and Definitions}\label{sec:preliminary}

\ifdefined\ARXIV
In Table~\ref{table:symbols}, summary of notations and symbols used throughout
this paper are described.

\begin{table}[h!]
	\caption{Summary of symbols used in this paper.}\label{table:symbols}
	\centering
	\resizebox{0.95\columnwidth}{!}{
		\bigskip
		\begin{tabular}{| c | c |}
			\hline
			\textbf{Symbols} & \textbf{Description} \\
			\hline
			\hline
			\Dset & original dataset of points\\
			\hline
			\DNset & neighbour dataset with \Dset, differing in one record\\
			\hline
			\QRset & set of query regions \\
			\hline
			t, \qr{t} & specification of a query region, including shape, size and position, $t \in $ \QRset\\
			\hline
			\Cset & set of cells \\ %overlapping a \QR\\
			\hline
			\c{i} & count value for the $ i $-th cell component, $i \in$ \Cset \\
			\hline
			\d{i} & true number of data/points in cell $ i $, within \QR, i $\in$ \Cset 
			\\
			\hline
			\a{i} & fraction of the overlapping area of \QR with cell $i$, \a{i} $\in (0,1]$
			\\
			\hline
			\Gset & set of grid sizes, \g{r}\\
			\hline
			\g{r} & size of a grid, number of divisions on each direction, $ r \in \Gset$\\
			\hline
			$ \rho $ & sanity bound for the relative error, computed for a dataset
			\Dset\\
			\hline
			$\Y{i}$ & noise added to cell $i$\\
			\hline
			$\lambda$ & scale parameter of Laplace mechanism\\
			\hline
			$\delta$ & a small value in $(0,1)$, used for the sanity bound  \\
			\hline
			$\epsilon, \epsilon_1, \epsilon_2$ & privacy parameters  \\
			\hline
			$c$ & constant in \heuristic approach \\
			\hline
	\end{tabular}}
\end{table}
\else
We now summarise background concepts and introduce notation.
\fi

\subsection{Spatial Data Structures}
\ifdefined\ARXIV
As discussed in Section~\ref{sec:related_work}, there is a wide range of spatial
data structures~\cite{samet2006} proposed for spatial object, from points,
path trajectories, to planar regions (bodies). \fi
Our focus on spatial histograms derives from their wide popularity in supporting
aggregate range queries. Originally developed for efficiency, histograms have
found application in qualitative privacy~\cite{Tao2004,Chawla2005}. 
Consider a dataset of points (locations), \Dset, where each record is a point.
Figure~\ref{fig:grid_histogram} displays a grid data structure of
points (Figure~\ref{fig:grid}) and the resulting spatial histogram \Hb of
counts $c_i$ per cell $i\in\Cset$ the set of cells (Figure~\ref{fig:histogram}). An aggregate range query is 
represented by a query region \QR (a red bolded rectangle in
Figure~\ref{fig:grid_histogram}), with corresponding responses as an
approximate count of points of \Dset that fall in that query region. We 
apply the \term{uniformity assumption} to quantify the contribution of a cell as the
cell count multiplied by the fraction of cell area in \QR (\cf
Section~\ref{sec:approach}).

\begin{figure}[h!]
	\centering
	\subfloat[Grid \label{fig:grid}]{
		\includegraphics[width = .3\columnwidth]{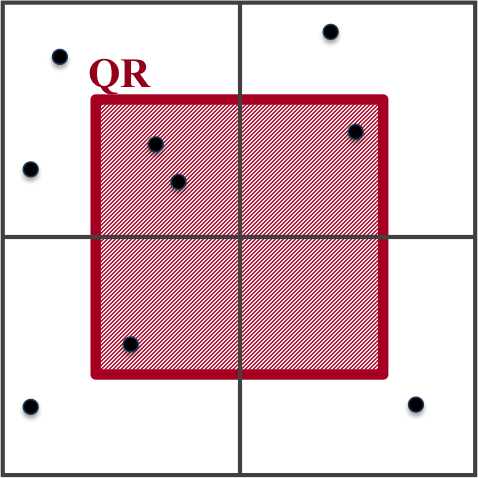}}
	\hspace{.4cm}
	\subfloat[Histogram \label{fig:histogram}]{
		\includegraphics[width = .3\columnwidth]{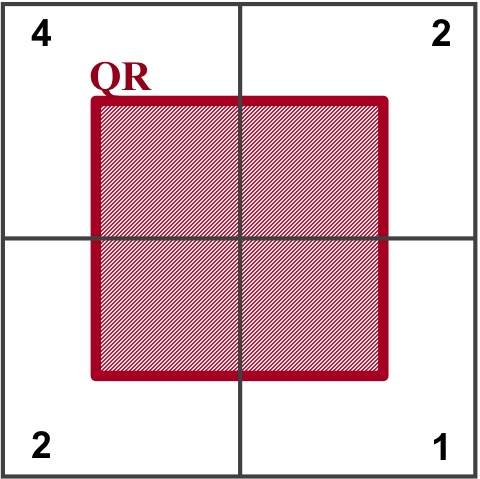}}
	\caption{Points, a $2 \times 2$ grid, and corresponding 
		histogram.} 
	\label{fig:grid_histogram}
\end{figure}

\subsection{Differential Privacy}

We adopt the differential privacy (DP)~\cite{Cynthia2006,privacybook14} framework
due to its strong guarantees on data privacy. 

\begin{definition} %\term{(Neighbouring)}
	Databases \Dset and \DNset that differ on exactly one record, with
	\DNset having one more than \Dset, are termed \term{neighbours}.
	%	Example: \Dset could represent a set of points.
\end{definition}

\begin{definition}%[\cite{Cynthia2006}]
	\label{def:dp}
	A randomised mechanism $\M$, preserves \term{$\epsilon$-differential
		privacy} for $\epsilon>0$, if for all neighbouring databases $\Dset, \DNset$
	and measurable $S \subseteq Range(\M)$:
	\begin{eqnarray*}
		\Prob{\M(\Dset)\in S} &\leq& \exp(\epsilon) \cdot \Prob{\M(\DNset)\in S}
		\enspace.
	\end{eqnarray*}
\end{definition}
Differential privacy requires that small changes to input (addition/deletion
of a record) do not significantly affect a mechanism's response distribution.
As such sampling from the mechanism's output cannot be used to distinguish the
input database.

\begin{lemma}[\cite{Cynthia2006}] %\term{(Composition~\cite{McSherry07,McSherry09})}
	\label{lem:seq}
	Consider mechanisms $\M_i$ each providing $\epsilon_i$-differential privacy,
	then the release of the vector of mechanism's responses on database \Dset preserves
	$\sum_{i}\epsilon_i$-differential privacy.	
\end{lemma}

\begin{definition} %\term{(Sensitivity, $\Delta$)}
	The \term{$L_1$-global sensitivity (GS)} of a
	deterministic, Euclidean-vector-valued function $f$ is given by $\Delta f =
	\max\limits_{\Dset,\DNset} \norm{f(D) -f(D')}_1$, taken over neighbouring databases.
\end{definition}

The simplest generic mechanism for differential privacy smooths non-private function 
sensitivity with additive perturbations.

\begin{theorem}[\cite{Cynthia2006}]\label{theorem:lap}
For any deterministic Euclidean-vector-valued $f(D)$, the \term{Laplace mechanism}
$M(\Dset)\sim Lap(f(\Dset), \Delta f / \epsilon)$ preserves $\epsilon$-differential
privacy.
\end{theorem}

Another important mechanism enables release from arbitrary sets that need not be 
numeric. 

\begin{theorem}[\cite{McSherry07}] \label{theorem:exp}
Consider a \term{score function} (or quality, utility function)
$s(\Dset, r)\in\mathbb{R}$ for database \Dset and response $r\in\mathcal{R}$. 
Then the \term{exponential mechanism} that outputs response $r$ with probability
\begin{eqnarray}
\Prob{M(s, \Dset) = r} = \frac{\exp(\epsilon \cdot s(\Dset,
	r)/2\Delta)}{\sum_{r \in \mathcal{R}} \exp(\epsilon \cdot s(\Dset, r)/2\Delta)}\enspace,
\label{eq:exp}
\end{eqnarray}
preserves $\epsilon$-differential privacy for $\epsilon>0$ and $\Delta=\Delta s$.
\end{theorem}

The exponential mechanism is typically used with
$\Delta s=\sup_{r\in\mathcal{R}} \Delta s(\cdot, r)$. However, using response 
dependent sensitivity per term achieves the same privacy, with potentially better utility:

\begin{definition} %\term{(Response Dependent Sensitivity, \GS{r})}
	\term{Response-dependent sensitivity} is $\GS{r}=\Delta s(\cdot, r)$.
\end{definition}

%% file: subfiles/problem.tex
\section{Problem Statement}\label{sec:problem}

We seek to address the problem of parameter tuning 
spatial histograms in an end-to-end differentially-private setting
(\cf Figure~\ref{fig:E2E_process}).

\begin{prob}\label{prob:main}
Given point-set \Dset, a set of query regions \QRset, budget
$\epsilon>0$, our goal is to batch process \Dset to produce a data
structure that can respond to an unlimited number of range queries
through privately selecting a grid size from given set \Gset that optimises response
accuracy on queries \QRset, while preserving $\epsilon$-differential
privacy.
\end{prob}

\subsection{Evaluation Metrics}\label{sec:metrics}

Specifically, solutions should have the following properties:

\begin{property}[End-to-End Differential Privacy] 
	Mechanisms should achieve non-interactive
	differential privacy not only in the release of a data
	structure based on spatial data but also in parameter
	tuning \eg grid size selection, of the structure.
\end{property}

\begin{property} [Utility: Low Relative Error]
	Mechanisms should achieve low total error on future
	query regions \QR, as measured by relative error
	$\left|response(\QR) - true(\QR)\right|/true(\QR)$. 
\end{property}

\begin{property} [Efficiency: Low Computational Complexity]
	Mechanisms should enjoy low computational time complexity
	in terms of key parameters of the data and geographic area.
\end{property}

\myparagraph{Error trade-off.}
We expect a trade-off between two sources of error as depicted
in Figure~\ref{fig:errors}, illustrating the need to tune grid 
size: aggregation error due to failure of the uniformity 
assumption when aggregating for qualitative privacy; perturbation 
error due to count noise introduced for differential privacy.
\begin{figure}[t]
	\centering
	\includegraphics[scale = .26]{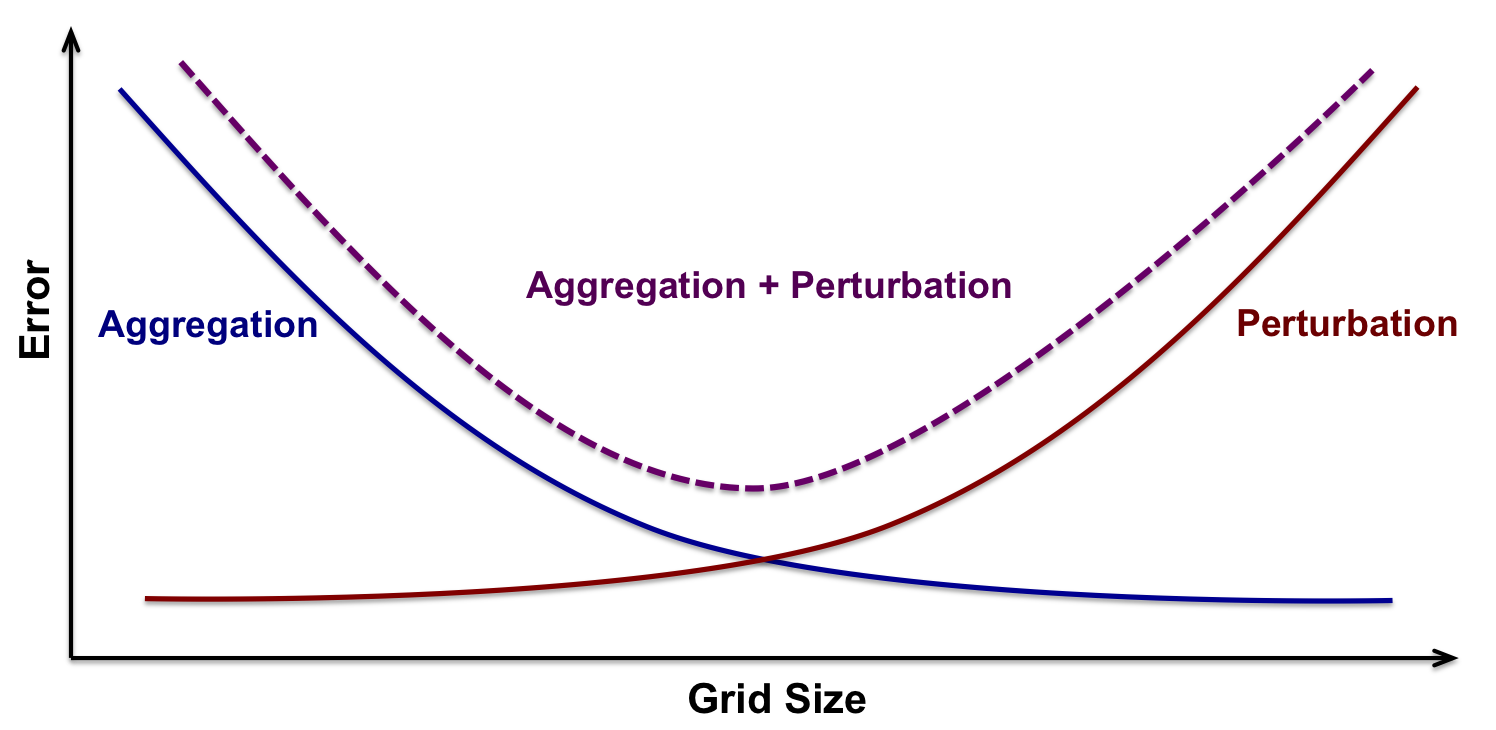}
	\caption{Expected error trade-off, demonstrating that 
	optimal grid size depends on the combination of error 
sources.}\label{fig:errors}
\end{figure}

\section{Commentary on Heuristic Approach}\label{sec:heuristic}

Qardaji~\etal~\cite{QardajiYL13} propose the \heuristic grid size
selection approach as the fixed-rule Equation~\eqref{eq:heuristic}. 
An idealisation of the kind of situation in which \heuristic fails
is presented in Figure~\ref{fig:heuristic_problem}. \heuristic
might suggest a $4 \times 4$ grid here (Figure~\ref{fig:heuristc1})
based on the number of points and assuming uniformity. However, a
\QR that happens to be located over regions of
non-uniformity---precisely where the uniformity assumption
fails--leads to erroneous query response. For concreteness, if the
four well-populated cells contain 100 points each (with just 1 each
within the \QR), then the response on the \QR would be 100: each
cell contributes $100\cdot 0.25$.
By contrast, on an alternate $8\times 8$ partitioning
(Figure~\ref{fig:heuristic2}), the response to the same \QR would
be the correct count of $4$. In this case, the uniformity assumption
and \QR align perfectly. \heuristic is derived with 
reliance on the uniformity assumption, and is incapable of
adapting to datasets where it holds to a greater/lesser degree.

\begin{figure}[b]
	\centering
	\subfloat[$4 \times 4$ grid \label{fig:heuristc1}]{
		\includegraphics[width = .4\columnwidth]{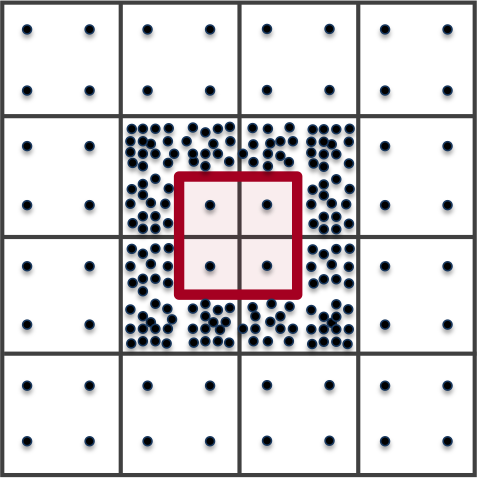}}
	\hspace{.2cm}
	\subfloat[$8 \times 8$ grid \label{fig:heuristic2}]{
		\includegraphics[width = .4\columnwidth]{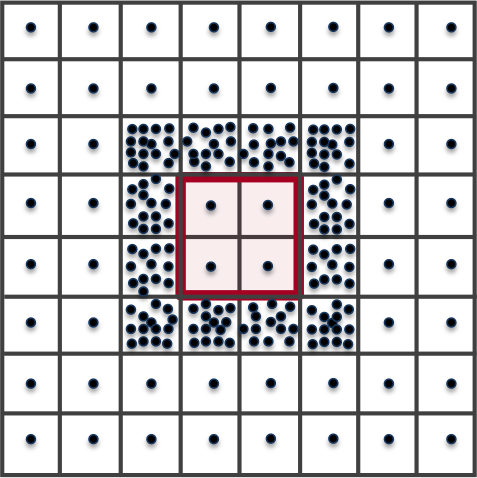}}
	\caption{The bolded red rectangle depicts the Query Region (\QR), dots illustrate points. The non-uniformly located points across the \QR boundary resulting in erroneous count using uniformity assumption for a $4 \times 4$ grid.} 
	\label{fig:heuristic_problem}
\end{figure}

While the derivation of \heuristic considers both error sources 
separately, the combination of bounds is not justified. 
Our approach privately optimises a rigorously-derived bound on
total error.

Finally, the recommendation \term{c}$=10$ is determined not 
on unrelated datasets, but openly optimised utility on the
\emph{evaluation} datasets. Not only does this practice
violate differential privacy~\cite{Hay16} but it fails to guarantee
good utility when applied to future datasets. 

%\subsection{Failure to Trade-off Error in the Real-World}

Motivated by the expected need for balancing errors
through data-dependent grid tuning (\cf
Figure~\ref{fig:errors}), we explore utility vs. grid size  
under the Storage dataset for fixed \QR of 1\% of
domain size (\cf Section~\ref{sec:experiment} for dataset details).
This dataset was used in the non-private tuning of $c$ in \heuristic in
\cite{QardajiYL13}. While the results of the tuning where not compared
with the true optimum, we make this comparison in
Figure~\ref{fig:grid_error}. The trade-off between errors is as predicted
(Figure~\ref{fig:errors}). Moreover the grid chosen by \heuristic is far
from optimal, further confirming that fixed parameters are unsuitable
for accurate responses, and that private data-dependent tuning is needed.

\begin{figure}[t]
	\centering
	\includegraphics[scale = .25]{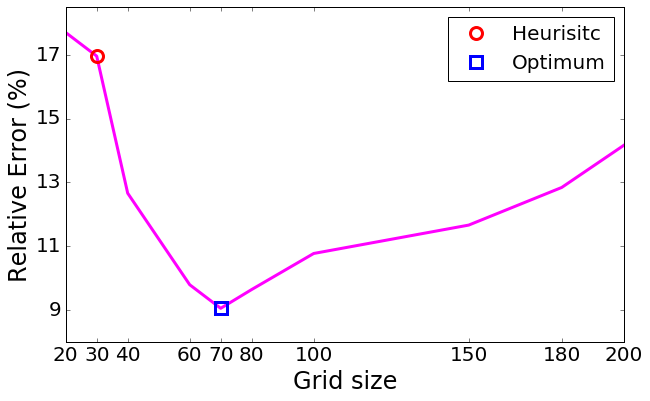}
	\caption{Effect of grid size on response utility for Storage dataset, demonstrating selection must be data dependent.}\label{fig:grid_error}
\end{figure}

%% file: subfiles/approach.tex
\section{Approach: \EE}\label{sec:approach}
Our solution to end-to-end $\epsilon$-differentially-private
histogram release, \EE, consists of two phases with budgets
$\epsilon_1+\epsilon_2=\epsilon$: 
1) 
Select one from a set of given grid sizes, by
privately minimising data-dependent expected error bounds on
a given set of {\QR}s;
2)
Construct a histogram with chosen grid size, privatized
by perturbing cell counts. Algorithm~\ref{algo:private-grid}
(Section~\ref{sec:private-grid}) describes these phases and the
process of responding to subsequent queries using the released
data structure is described by Algorithm~\ref{algo:range}
(Section~\ref{sec:range-query}).

\subsection{Spatial Histograms Release}\label{sec:private-grid}

Consider Algorithm~\ref{algo:private-grid}, which releases a
tuned spatial histogram. In Phase~1
[lines \ref{algo:private-grid/phase1-start}--\ref{algo:private-grid/phase1-end}],
a histogram \Hb
is constructed on \Dset, for each candidate grid size in \Gset.
In Phase~2, cell counts will be privatized by adding
Laplace-distributed r.v. $Y_i$ to count $i$ as
\begin{eqnarray}
%\Y{i}: \mbox{noise added to cell i}\enspace, \nonumber\\ 
\Y{i} \sim Lap(0; \lambda)\ ,\ \ \ Var(\Y{i}) = 2\lambda ^ 2\ ,\ \ \ \Exp{\left|\Y{i}\right|} = \lambda \label{eq:noise}\enspace, 
%\left|\Y{i}\right|\sim Exp\left(\frac{1}{\lambda}\right) \Rightarrow %
 %= \frac{1}{\epsilon_1} 
\end{eqnarray}
taking\footnote{Since the global sensitivity for histogram release is 1.}
$\lambda=1/\epsilon_2$. The idea behind the 
algorithm is to compute a bound on the expected relative
error that this (future) noisy histogram would incur, averaged
over the {\QR}s in \QRset, as evaluated on the data \Dset.
The bound's expression (Corollary~\ref{cor:gs_all}) involves
comparing the histogram \Hb response on each \QR with the true
count on \QR, and then to the absolute of this quantity 
(reflecting aggregation error) the expected perturbation
\eqref{eq:noise}.

Note that histogram response to \QR involving the uniformity
assumption requires computation of the area overlap between
each cell $i$ in \Hb and the \QR,
\begin{eqnarray}
\a{i} = \frac{Area(QR\cap cell_{i})}{Area(cell_{i})} \in [0,1] ,\enspace i \in \Cset\enspace. \label{eq:alpha} 
\end{eqnarray}

\IncMargin{1em}

\begin{algorithm}
	\LinesNumbered
	\SetKwInOut{Input}{Input}\SetKwInOut{Output}{Output}
	
	\Input{Dataset of points \Dset; Set of grid sizes \Gset; Set of range queries \QRset; Privacy budgets $\epsilon_1,\epsilon_2$}
	\Output{Selected grid size \gstar; Histogram \Hn}
	\BlankLine
	\textbf{Phase 1:} \textit{Privately tune the grid cell size:} \label{algo:private-grid/phase1-start}
	
	\For{$g \in \Gset$}{

		Initialise histogram \Hb of counts per cell (each sized $g$).

		\For{$p \in \Dset $}{
			Increment the count of \Hb's cell that $p$ falls in.
		}
		
		\For{$\QR \in \QRset$}{
			Compute overlap \a{i} between \QR \& each cell $i$ in \Hb.

			Compute {\Hb}'s non-private \QR response $response(\QR)=\sum_i \a{i} \cdot c_i$.

			Count $true(\QR)$ number points in \Dset falling in \QR.

			Compute $error(\QR)$ bound as per Corollary~\ref{cor:gs_all}
		}
		Compute average error bound $avgError$ over \QRset.

		Set exponential mechanism score $s(D,g) = -avgError$.
		
		Compute the probability of responding $g$, using privacy parameter $\epsilon_1$, as Equation~\eqref{eq:exp}

		Compute bound/score sensitivity $\Delta_g$ from Corollary~\ref{corollary:gs_relative_avg}.
	}		
	Sample $\gstar$ as $r$ w.p. $\propto \exp(\epsilon_1 \cdot s(\Dset, r)/2\Delta_r)$. \label{algo:private-grid/phase1-end}

	\textbf{Phase 2:} \textit{Construct the private histogram counts:} \label{algo:private-grid/phase2-start}
	
	Re-create the histogram \Hb for chosen grid size \gstar.
	
	Perturb the cell counts with iid Laplace noise per cell, 
	$\Hn = \Hb + \Yb$, $\Yb \sim Lap(0; \lambda), \lambda = 1/\epsilon_2$. \label{algo:private-grid/phase2-end}

	\caption{End-to-End Differentially Private Spatial Histogram Construction via Private Grid Size Tuning}\label{algo:private-grid}
	
\end{algorithm}
\DecMargin{1em}

To minimise averaged error bound (over each query in \QRset)
we set the exponential mechanism's score function (\cf
Theorem~\ref{theorem:exp}) for maximisation to be the negative
error. To calibrate the mechanism, 
we use the sensitivity of this score function as bounded
in Corollary~\ref{corollary:gs_relative_avg}. Detailed derivation
of these bounds is provided in Section~\ref{sec:analysis}.
The result is a sampled \gstar which 
approximates the grid size optimising the (data-dependent
non-private) error bound.

In Phase~2
[lines \ref{algo:private-grid/phase2-start}--\ref{algo:private-grid/phase2-end}],
a private histogram is produced
for the chosen grid size \gstar using the Laplace
mechanism---following the same process as simulated in
Phase~1.\\[-0.5em]

\myparagraph{Computational Complexity.}
Algorithm~\ref{algo:private-grid} is efficient with time
complexity 
$O(\abs{\Dset}\cdot\abs{\QRset} \cdot\abs{\Gset} \cdot g^2)$
and space complexity
$O(\abs{\Gset} + g^2)$. The parameter
$g$ is the largest grid size
in \Gset: it is necessary to touch at least every cell.

\ifdefined\ARXIV

\subsubsection{Computing cell, \QR overlap}\label{sec:QR_ratio}

Figure~\ref{fig:grid-QR} illustrates an example \QR intersecting with a histogram cells.

\begin{figure}[h]
	\centering
	\includegraphics[scale = .35]{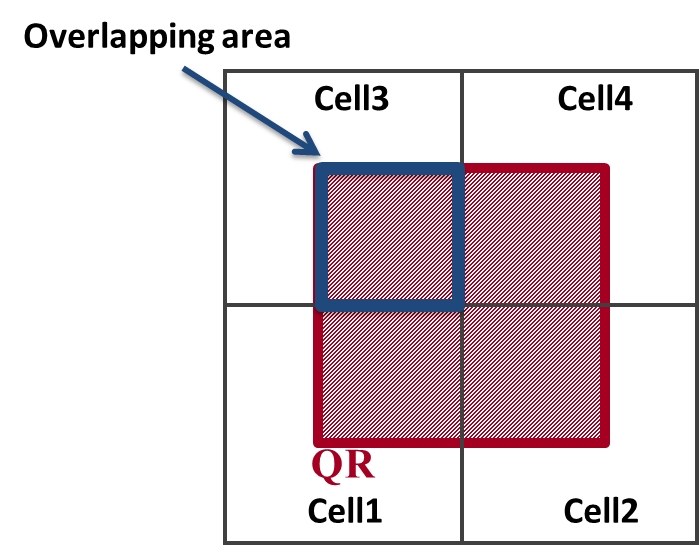}
	\caption{A query region, \QR, intersecting with grid cells and its overlapping area of a cell.}\label{fig:grid-QR}
\end{figure}

To compute the fraction \a{i} of overlapping area of a \QR
with a cell $i$ as defined in Equation~\eqref{eq:alpha}, we
have the special case of polygon intersection from
computational geometry~\cite{ORourke1998}
\begin{eqnarray*}
	Area(QR\cap cell_{i}) = x_{overlap} \times y_{overlap}
\end{eqnarray*}
where
\begin{eqnarray*}
	x_{overlap} &=& \max(0, \min(x_{12}, x_{22}) - \max(x_{11}, x_{21})) \\
	y_{overlap} &=& \max(0, \min(y_{12}, y_{22}) - \max(y_{11}, y_{21})) \\
	\mbox{cell coordinates} &=& [(x_{11}, y_{11}),(x_{12}, y_{12})]\enspace,  \\
	\mbox{\QR coordinates} &=& [(x_{21}, y_{21}),(x_{22}, y_{22})]\enspace.  \\	
\end{eqnarray*}
Note that
$\sum_{i \in \Cset} \ceil{\a{i}} \equiv \mbox{ \# overlapping cells}$.

\fi

\subsection{Post-Release Range Query Response}\label{sec:range-query}

\IncMargin{1em}

\begin{algorithm}
	\LinesNumbered
	\SetKwInOut{Input}{Input}\SetKwInOut{Output}{Output}
	
	\Input{Histogram \Hb with cells \Cset; Query \QR}
	\Output{Approximate count $response(\QR)$}
	
	\For{$i \in \Cset$}{
		Compute overlap \a{i} of $i$ with \QR, using Equation~\eqref{eq:alpha}.

		Compute cell's contribute $z_i = \a{i} \cdot c_i$.
	}
	Sum the contributions $response(\QR) = \sum_{i\in\Cset} z_i$.
	\caption{Range Query Response}\label{algo:range}
	
\end{algorithm}
\DecMargin{1em}

Algorithm~\ref{algo:range} takes histogram \Hb, query \QR, to a response. 
The algorithm simply weighs each cell's count by its overlap
\a{i} with the \QR, applying the uniformity assumption.\\[-0.5em]

\myparagraph{Computational Complexity.} The range query algorithm
is efficient in time $O(g^2)$ linear in the number of cells
\& constant space.

%% file: subfiles/analysis.tex
\section{Theoretical Analysis}\label{sec:analysis}

Having described key concepts underlying Algorithm~\ref{algo:private-grid}
in the previous section, we now derive the bound on expected error 
of Phase 2's histogram release (Corollary~\ref{cor:gs_all}),
that is privately minimised by the mechanism; we prove differential privacy
(Theorem~\ref{theorem:exp-privacy}) and provide a utility bound
(Theorem~\ref{theorem:exp-utility}). A key component of our analysis
is in bounding sensitivity of our error bound to perturbations in the
input dataset (Corollary~\ref{corollary:gs_relative_avg}). By using a more
refined response-dependent sensitivity our mechanism enjoys improved
utility at no price to privacy (\cf Section~\ref{sec:sensitivity-results}
for a discussion).

We begin our analysis for the single tuning query case
(Section~\ref{sec:single-qr}, and then extend to multiple
queries (Section~\ref{sec:multiple-qr}).

\subsection{Case: Single Tuning Query}\label{sec:single-qr}
% \myparagraph{Expected Error.} %Ben: I'm not sure we need this level of structure. Open to suggestion otherwise!
We first bound expected error of Phase 2 when responding to a single
(tuning) \QR. We bound both absolute error, and relative
error. We introduce a constant $\rho$ in the denominator of the latter
in order to control sensitivity in Theorem~\ref{theorem:gs_relative},
as discussed in Remark~\ref{rem:sanity}.

\begin{theorem}\label{theorem:qr_error}
For any given query region \QR, the histogram \Hn released by
Algorithm~\ref{algo:private-grid} Phase~2 on data \Dset achieves
expected error (wrt randomness in the \Y{i}) bounded as, 
\begin{enumerate}[(i)]
	\item Absolute error:
\begin{eqnarray*}
	\Exp{\left|response\left(\QR\right) - true\left(\QR\right)\right|} &\le& \left|\sum_{i \in \Cset}\a{i} \c{i} - \sum_{i \in \Cset} \d{i}\right| + \lambda \norm{\Alpha}_{1}.
\end{eqnarray*}
\item Relative error:
\begin{eqnarray*}
	 \Exp{\frac{\left|response(\QR) - true(\QR)\right|}{\max\{true(\QR), \rho\}}} 
	\le  \frac{\left|\sum_{i \in \Cset}\a{i} \c{i} - \sum_{i \in \Cset} \d{i}\right| + \lambda \norm{\Alpha}_{1}}{\max\{\sum_{i \in \Cset} \d{i}, \rho\}} .
\end{eqnarray*}
\end{enumerate}
where $\rho>1$ is a constant (\cf Remark~\ref{rem:sanity}), and \d{i}
counts the number of points in \Dset falling in both cell $i$ and \QR.
 \end{theorem}

 \begin{proof}
 	Consider the first case of absolute error,
 	\begin{eqnarray*}
 		%\Exp{error(\QR)} &=& \Exp{\left|response(\QR) - true(\QR)\right|}\\
 		%&=&
		&& \Exp{\left|\sum_{i \in \Cset} \a{i}(\c{i} + \Y{i}) - \sum_{i \in \Cset} \d{i}\right|}\\
 		%%&=& \Exp{\left|\left(\sum_{i \in \Cset} \a{i}\c{i} - \sum_{i \in \Cset} \d{i}\right) +  \sum_{i \in \Cset} \a{i} \Y{i}\right| }\\
	 &\le& \Exp{\left|\sum_{i \in \Cset}\a{i} \c{i} - \sum_{i \in \Cset} \d{i}\right|} + \Exp{\left|\sum_{i\in\Cset} \a{i} \Y{i}\right|}\\
% 		&=& \underbrace{\left|\sum_{i \in \Cset}\a{i} \c{i} - \sum_{i \in \Cset} \d{i}\right|}_{\mbox{aggregation error}} + \underbrace{\Exp{\left|\Alpha \cdot \Yb\right|}}_{\mbox{perturbation error}}\\
 		&\le& \left|\sum_{i \in \Cset}\a{i} \c{i} - \sum_{i \in \Cset} \d{i}\right| + \lambda \norm{\Alpha}_{1}\enspace,\\
% 		&\le& \left|\sum_{i \in \Cset}\a{i} \c{i} - \sum_{i \in \Cset} \d{i}\right| + \lambda  \sum_{i \in \Cset} \ceil{\a{i}}\enspace.
 	\end{eqnarray*}
 	where the first inequality follows from rearranging terms
	and applying the triangle inequality and monotonicity \& linearity
	of expectation; the second inequality follows from the same arguments
	combined with Equation~\eqref{eq:noise}:
 	\begin{eqnarray*}
 	\Exp{\left|\Alpha\cdot \Yb\right|}
	\ \ \le\ \  \sum_{i \in \Cset} \a{i}\Exp{\left|\Y{i}\right|} 
% 		&=& 	\sum_{i \in \Cset} \a{i}\lambda = \lambda \sum_{i \in \Cset} \a{i} =
	 \ \ =\ \ \lambda \norm{\Alpha}_{1}\enspace.\\
% 		\lambda \norm{\Alpha}_{1} &\leq& \lambda  \sum_{i \in \Cset} \ceil{\a{i}}\enspace.\\
	\end{eqnarray*}
%% We had this above already. Plus, we're not using it, so it shouldn't appear in the proof. Basically it's a bit circular. For the sum of d's is there b/c it is the true count. So there's not much need to point it out again at the end. Yes, it is insightful, so better in main text.
%and  
%\begin{eqnarray*}
%	\sum_{i \in \Cset} \d{i} \equiv \mbox{ \# true data in \QR}\enspace .
%\end{eqnarray*}	
 The second claim follows immediately.
 \end{proof}

 \begin{rem}
	 It is notable that the bound decomposes total (expected) error
	 into two interpretable terms:
	 $\left|\sum_{i \in \Cset}\a{i} \c{i} - \sum_{i \in \Cset} \d{i}\right|$
	 reflecting aggregation error due to spatial aggregation and
	 (potential) failure of the uniformity assumption; and
	 $\lambda \|\Alpha\|_1$
	 reflecting error due to random perturbation from the Laplace
	 mechanism, where $\lambda$ is noise scale and $\|\Alpha\|_1$
	 counts the (effective) cells overlapping the \QR.
 \end{rem}

% We introduce sanity bound to control the sensitivity of the score function.
\begin{rem}\label{rem:sanity}
	$ \rho>1 $ is a user-defined constant, referred to as the
	\term{sanity bound} in the literature~\cite{GarofalakisG02,VitterW99}.
	It is commonly used to control sensitivity of relative error measures
	in the face of small true counts that can potentially yield unbounded
	blow-up of relative error. Previous recommendations set it as 
	$\rho = \delta\times|\Dset|$, where
	$0<\delta<1$ is taken to be a small constant reflecting a
	pseudo-count fraction of \Dset.
\end{rem}

As Algorithm~\ref{algo:private-grid} Phase~1 privately minimises the
relative the error bound on Phase~2 of Theorem~\ref{theorem:qr_error}---using
the exponential mechanism---we must compute the sensitivity of this
bound which itself is data-dependent and hence privacy-sensitive. 
\emph{We cannot simply optimise the error bound of
Theorem~\ref{theorem:qr_error} directly, as implicitly done
by \heuristic, lest we breach data privacy.}

% Again, think we don't need the extra sectioning.
%\myparagraph{Score Function.}
We define the exponential mechanism's
score (quality) function as the negative relative error bound: maximising
this score over candidate grid sizes \Gset, equivalently 
minimised the error bound, which in turn is a close surrogate for 
minimising actual future error of Phase~2 on the tuning query set
\QRset.
\begin{eqnarray}
%score(data, response)& = & s(\Dset,r) \nonumber\\
s(\Dset,r) & = & - \frac{\left|\sum_{i \in \Cset}\a{i} \c{i} - \sum_{i \in \Cset} \d{i}\right| + \lambda \norm{\Alpha}_{1}}{\max\{\sum_{i \in \Cset} \d{i}, \rho\}}\enspace. \label{eq:score}
\end{eqnarray}

\ifdefined\ARXIV
And we make the analogous definition if optimising absolute error:
\begin{eqnarray*}
s(\Dset,r) & = & - \left|\sum_{i \in \Cset}\a{i} \c{i} - \sum_{i \in \Cset} \d{i}\right| - \lambda  \norm{\Alpha}_{1}\enspace.
\end{eqnarray*}
\fi

To calibrate the exponential mechanism for differential
privacy, we must bound the sensitivity $\Delta s$ of the score function.

\ifdefined\ARXIV
\begin{lemma}\label{lem:gs}
	The global sensitivity of the absolute score function, is bounded above by $\left|1- \a{i}\right|$ which is at most $1$, as each $\a{i} \in (0,1]$.
\end{lemma}
\begin{proof}
	From the reverse triangle inequality we have
	\begin{eqnarray*}
		&& GS(s) = \Delta= \max_{r, \norm{\Dset - \DNset}_1 \le 1} \abs{s(\Dset,r) - s (\DNset, r)}\\
		&\le& \left|\left(\left|\sum_{i \in \Cset}\a{i} \c{i} - \sum_{i \in \Cset} \d{i}\right| + \lambda  \norm{\Alpha}_{1}\right) - \left(\left|\sum_{i \in \Cset}\a{i} c'_{i} - \sum_{i \in \Cset} d'_{i}\right| + \lambda  \norm{\Alpha}_{1}\right)\right|\\
		&\le& \abs{\abs{\a{1}\c{1} + \a{2}\c{2} + \dots + \a{i}\c{i} + \dots - (\d{1} + \d{2} + \dots + \d{i} + \dots)} - \\
			&& \abs{\a{1}\c{1} + \a{2}\c{2} + \dots + \a{i}(\c{i} + 1) + \dots \\
				&& - (\d{1} + \d{2} + \dots + \d{i}+1 + \dots)}}\\	
		&\le& \abs{\abs{\a{1}\c{1} + \a{2}\c{2} + \dots + \a{i}\c{i} + \dots - (\d{1} + \d{2} + \dots + \d{i} + \dots)} \\
			&& - \abs{\big(\a{1}\c{1} + \a{2}\c{2} + \dots + \a{i}\c{i} + \dots - (\d{1} + \d{2} + \dots + \d{i} + \dots)\big) \\
				&& + (\a{i} - 1)}}\\
		&=& \left|1- \a{i}\right|\enspace. %\Longrightarrow  \Delta \le \left|1- \a{i}\right| < 1 \enspace.
	\end{eqnarray*}
%The above follows by the reverse triangle inequality rule:\\
%$\abs{x-y} \ge \left|\abs{x} - \abs{y}\right|\enspace.$
\end{proof}

The case for relative error is much more involved.
\fi

\begin{theorem}\label{theorem:gs_relative}
The response-dependent sensitivity of relative error score
function~\eqref{eq:score}, for any $r\in \Gset$ and fixed
query $t\in\QRset$, is bounded
\begin{eqnarray*}
	\GS{r} s &\leq& 
	%\frac{1}{\delta(\rho+1)} + \frac{\lambda  \norm{\Alpha^{r,t}}_{1}}{\rho(\rho+1)} + \frac{1}{\rho + \delta} = 
	\frac{1}{\delta(\delta\left|\Dset\right|+1)} + \frac{\lambda  \norm{\Alpha^{r}}_{1}}{\delta\left|\Dset\right|(\delta\left|\Dset\right|+1)} + \frac{1}{\delta\left|\Dset\right| + \delta}\enspace,
\end{eqnarray*}
where $\delta \in (0,1)$ defines sanity bound constant
$\rho = \delta\left|\Dset\right|$. We introduce superscript $r$ to \Alpha,
to highlight explicit dependence on $r\in\Gset$.
\end{theorem}

	To prove the result we must bound the quantity %We drop superscripts on $\a{i}$ for readability.
	\begin{eqnarray*}
		%&& GS(s_r) = \GS{r}= \max_{\norm{\Dset - \DNset}_1 \le 1} \left|s(\Dset,r) - s (\DNset, r)\right|\\
		&&\left|\frac{1}{\max\{\sum_{i \in \Cset} \d{i}, \rho\}} \left(\left|\sum_{i \in \Cset}\a{i} \c{i} - \sum_{i \in \Cset} \d{i}\right| + \lambda  \norm{\Alpha}_{1}\right) - \right. \\ && \left. \frac{1}{\max\{\sum_{i \in \Cset} d'_{i}, \rho^{\prime}\}} \left(\left|\sum_{i \in \Cset}\a{i} c'_{i} - \sum_{i \in \Cset} d'_{i}\right| + \lambda  \norm{\Alpha}_{1}\right)\right|\enspace,
	\end{eqnarray*}
	where $\rho = \delta \times\left|\Dset\right|$ and
	$\rho^{\prime} = \delta \times\left|\DNset\right| = \delta \times(\left|\Dset\right| + 1) = \rho + \delta$.
The proof proceeds by cases, based on where {\DNset}'s extra
point falls: 
outside \QR and cells overlapping \QR (Figure~\ref{fig:case1}); 
outside \QR, inside cells overlapping \QR (Figure~\ref{fig:case2}); or
inside \QR (Figure~\ref{fig:case3}).
The reader interested in the (technical) calculations for the full proof
of the theorem are referred to 
\ifdefined\ARXIV
Appendix~\ref{sec:appendix}. 
\else
our full report~\cite{techreport}. 
\fi

%		& \mbox{and} & \\
%		&&\sum_{i \in \Cset} \d{i} = \begin{cases}
%				0 \mbox{, }max\{\sum_{i \in \Cset} \d{i}, \rho\} =  \rho\\
%				\ge 1 \mbox{, }max\{\sum_{i \in \Cset} \d{i}, \rho\}  = \sum_{i \in \Cset} \d{i} \mbox{ or } \rho
%			\end{cases}\\
%		&&\sum_{i \in \Cset} d'_{i} = \begin{cases}
%			\sum_{i \in \Cset} \d{i}\mbox{, } max\{\sum_{i \in \Cset} d'_{i}, \rho + \delta\} = \sum_{i \in \Cset} \d{i} \mbox{ or } \rho + \delta  \\
%			\sum_{i \in \Cset} \d{i} + 1 \mbox{, } max\{\sum_{i \in \Cset} d'_{i}, \rho + \delta\} = \sum_{i \in \Cset} \d{i} + 1 \mbox{ or } \rho + \delta \\
%		\end{cases}

	\begin{figure}[b!]
		\centering
		\subfloat[Case 1.\label{fig:case1}]{
			\includegraphics[width = .3\columnwidth]{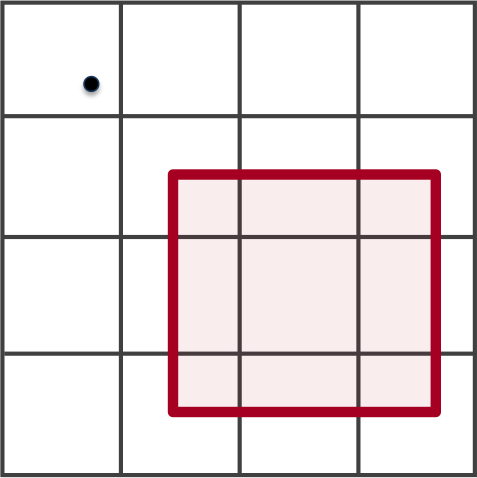}}
			\hspace{.15cm}
			\subfloat[Case 2.\label{fig:case2}]{
				\includegraphics[width = .3\columnwidth]{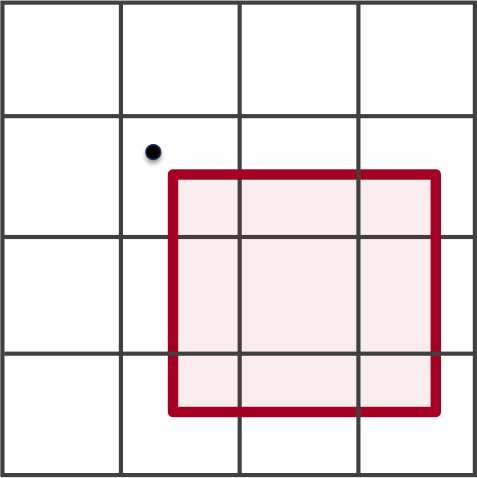}}
			\hspace{.15cm}
			\subfloat[Case 3.\label{fig:case3}]{
				\includegraphics[width = .3\columnwidth]{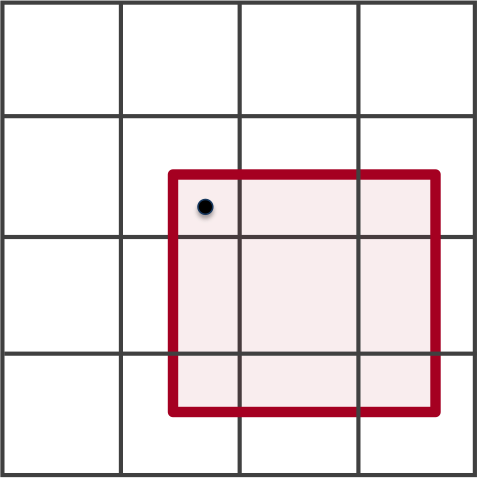}}
				\caption{Cases in the proof of Theorem~\ref{theorem:gs_relative} based on where the extra point (dot) falls relative to the \QR (red) and cells overlapping the \QR (the bottom right $3\times 3$ cells).} 
		\label{fig:cases}
	\end{figure}

%	\begin{figure}[h!]
%		\centering
%		\includegraphics[scale = 0.3]{figs/cases_caption}
%		\caption{Possible cases that extra point can be located with respect to the \QR and overlapping grid cells. Red line depicts the Query Rectangle (\QR) and dot illustrates extra point.} 
%		\label{fig:cases}
%	\end{figure}

\subsection{Case: Multiple Tuning Queries}\label{sec:multiple-qr}
Before proceeding to privacy and utility guarantees, we lift
the above single query analysis, to the case of multiple tuning
queries. The first step is bounding average
Phase~2 error over query set \QRset. This follows from 
Theorem~\ref{theorem:qr_error} and linearity of expectation.

\begin{corollary}\label{cor:gs_all}
For given set of query region \QRset, the histogram \Hn released by
Algorithm~\ref{algo:private-grid} Phase~2 on data \Dset achieves
average expected error (wrt randomness in the \Y{i}) bounded as, 
\begin{enumerate}[(i)]
	\item Absolute error:
\begin{eqnarray*}
	&& \Exp{\frac{1}{|\QRset|}\sum_{t\in\QRset}\left|response\left(t\right) - true\left(t\right)\right|} \\
	&\le& \frac{1}{|\QRset|}\sum_{t\in\QRset}\left(\left|\sum_{i \in \Cset}\a{i}^t \c{i} - \sum_{i \in \Cset} \d{i}\right| + \lambda \norm{\Alpha^t}_{1}\right)\enspace.
\end{eqnarray*}
\item Relative error:
\begin{eqnarray*}
&& \Exp{\frac{1}{|\QRset|}\sum_{t\in\QRset}\frac{\left|response(t) - true(t)\right|}{\max\{true(t), \rho\}}} \\
&\le&  \frac{1}{|\QRset|}\sum_{t\in\QRset}\frac{\left|\sum_{i \in \Cset}\a{i}^t \c{i} - \sum_{i \in \Cset} \d{i}\right| + \lambda \norm{\Alpha^t}_{1}}{\max\left\{\sum_{i \in \Cset} \d{i}, \rho\right\}} \enspace.
\end{eqnarray*}
\end{enumerate}
where $\rho>1$ is a constant (\cf Remark~\ref{rem:sanity}), \d{i}
counts the number of points in \Dset falling in both cell $i$ and \QR $t$, and
$\Alpha^t$ denotes the vector of cell overlaps with $t$.
\end{corollary}

For the general case, we therefore define the exponential mechanism's
score function as before, as the negative of the
bound on the expectation of the error averaged over \QRset,
\begin{eqnarray}
	s(\Dset,r) & = & - \frac{1}{|\QRset|}\sum_{t\in\QRset}\frac{\left|\sum_{i \in \Cset}\a{i}^t \c{i} - \sum_{i \in \Cset} \d{i}\right| + \lambda \norm{\Alpha^t}_{1}}{\max\{\sum_{i \in \Cset} \d{i}, \rho\}}\enspace. \label{eq:score-avg}
\end{eqnarray}

\ifdefined\ARXIV
And again we make the analogous definition if optimising absolute error:
\begin{eqnarray*}
	s(\Dset,r) & = & - \frac{1}{|\QRset|}\sum_{t\in\QRset}\left(\left|\sum_{i \in \Cset}\a{i}^t \c{i} - \sum_{i \in \Cset} \d{i}\right| + \lambda  \norm{\Alpha^t}_{1}\right)\enspace.
\end{eqnarray*}
\fi

We next extend the calculation of response-dependent sensitivity
of this bound to perturbations of the database \Dset.

\begin{lemma}\label{lem:gs_all_qr}
	For $i\in\mathcal{I}$ a finite index set, functions $f_i: \mathcal{X}\to\reals$ on arbitrary domain, and constants $\Delta_i\in\reals$,
	\begin{eqnarray*}
	\left(\forall i\in\mathcal{I}, \sup_{x\in\mathcal{X}} |f_i(x)| \leq \Delta_i\right) 
	&\Rightarrow& 
	\sup_{x\in\mathcal{X}}\left|\frac{1}{|\mathcal{I}|}\sum_{i\in\mathcal{I}} f_i(x)\right| \leq \frac{1}{|\mathcal{I}|}\sum_{i\in\mathcal{I}}\Delta_i
	\end{eqnarray*}
\end{lemma}

\begin{proof}
Applying the triangle inequality and distributing the supremum yields the result,
\begin{eqnarray*}
	\sup_{x\in\mathcal{X}}\left|\frac{1}{|\mathcal{I}|}\sum_{i\in\mathcal{I}} f_i(x)\right|
	&\leq& \sup_{x\in\mathcal{X}}\frac{1}{|\mathcal{I}|}\sum_{i\in\mathcal{I}} | f_i(x)| \\
	&\leq& \frac{1}{|\mathcal{I}|}\sum_{i\in\mathcal{I}} \sup_{x\in\mathcal{X}} | f_i(x)| \\
	&\leq& \frac{1}{|\mathcal{I}|}\sum_{i\in\mathcal{I}}\Delta_i\enspace.
\end{eqnarray*}
\end{proof}

\ifdefined\ARXIV
\begin{corollary}
The response-dependent sensitivity of averaged absolute error
score function over query set \QRset, for any
$r\in \Gset$, is bounded by $1$.
\end{corollary}

\begin{proof}
The claim bounds sensitivity of the absolute error score function 
derived from the averaged error bound of Corollary~\ref{cor:gs_all}.
The result follows immediately from Lemma~\ref{lem:gs_all_qr} by
taking: functions $f_i$ as the  sensitivities of the
individual \QR-specific score functions; and the
$\Delta_i$ bounds on each $f_i$ as the single-query sensitivity bound from
Lemma~\ref{lem:gs}.
\end{proof}
\fi

\begin{corollary}\label{corollary:gs_relative_avg}
The response-dependent sensitivity of averaged relative error
score function~\eqref{eq:score-avg} over query set \QRset, for any
$r\in \Gset$, is bounded
\begin{eqnarray*}
	\GS{r} s &\leq& 
	\frac{1}{\delta(\delta\left|\Dset\right|+1)} + \frac{(\lambda/|\QRset|)\sum_{t\in\QRset}  \norm{\Alpha^{r,t}}_{1}}{\delta\left|\Dset\right|(\delta\left|\Dset\right|+1)} + \frac{1}{\delta\left|\Dset\right| + \delta}\enspace,
\end{eqnarray*}
where $\delta \in (0,1)$ defines sanity bound constant
$\rho = \delta\left|\Dset\right|$. We introduce superscripts $r, t$ to \Alpha,
to highlight explicit dependence on $r\in\Gset, t\in\QRset$.
\end{corollary}

\begin{proof}
	The claim bounds sensitivity of the score function \eqref{eq:score-avg}
	derived from the
averaged error bound of Corollary~\ref{cor:gs_all}. The result follows immediately
from Lemma~\ref{lem:gs_all_qr} by taking: functions $f_i$ as the 
sensitivities of the individual \QR-specific score functions; and the
$\Delta_i$ bounds on each $f_i$ as the single-query sensitivity bound from
Theorem~\ref{theorem:gs_relative}.
\end{proof}

\begin{figure*}[t!]
	\centering
	\subfloat[Storage, with $8,938$ points.\label{fig:storage}]{
		\includegraphics[width = .3 \textwidth]{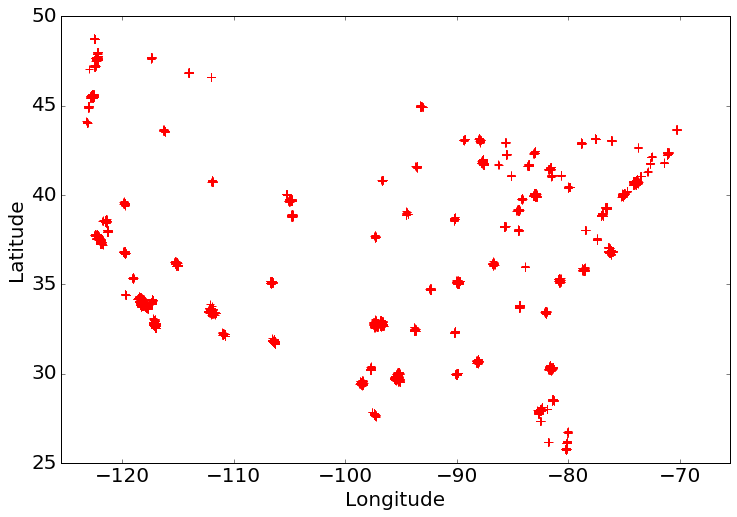}}
	\subfloat[Landmark, with $869,976$ points.\label{fig:landmark}]{
		\includegraphics[width = .3 \textwidth]{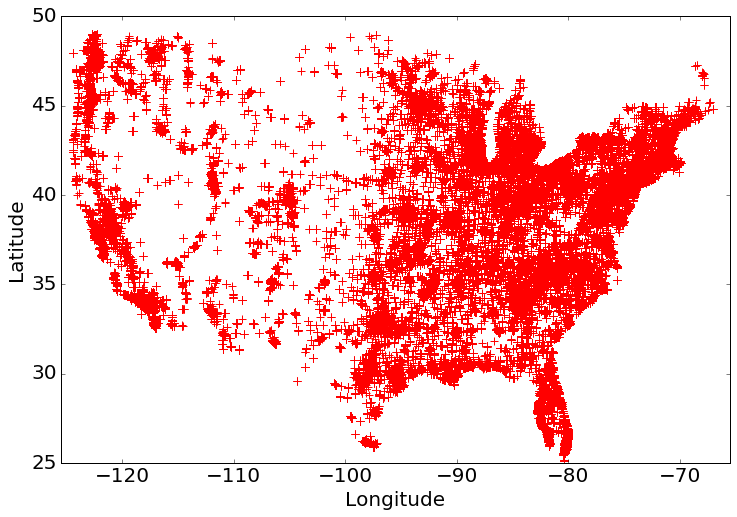}}
	\subfloat[Gowalla Check-ins, with $6,442,841$ points.\label{fig:gowalla}]{
		\includegraphics[width = .29 \textwidth]{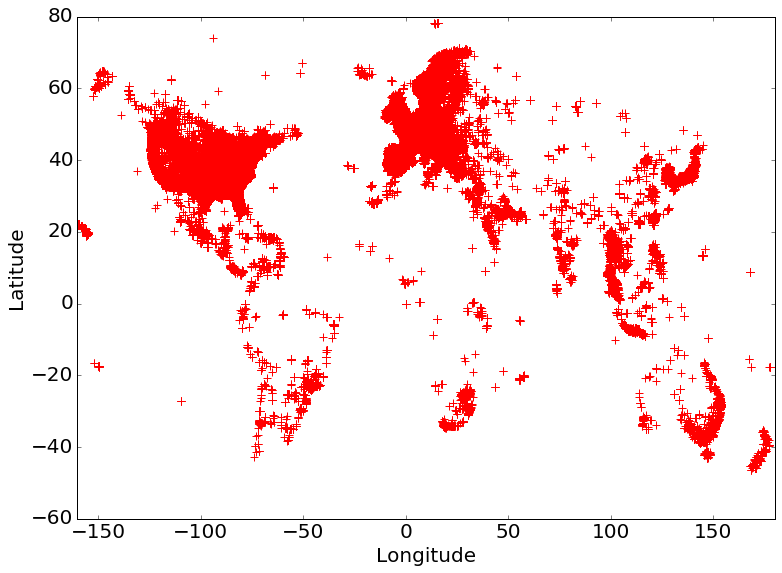}}
	\caption{Visualising the experimental datasets.}
	\label{fig:datasets}
\end{figure*}

\subsection{Main Results: Privacy \& Utility Guarantees}\label{sec:exp-utility}
With the Phase~2 error bounds and sensitivity of these bounds in hand, we are
able to present general guarantees for the end-to-end 
Algorithm~\ref{algo:private-grid}.

\begin{table*}[t!]
	\caption{Experimental settings. This table demonstrates the range of parameters, bolded are those that are varying.}\label{table:experiment}
	\centering
	\scalebox{.82}{
		\bigskip
		\begin{tabular}{| c | c | c | c | c | c | c |}
			\hline
			\textbf{Dataset, Size} & \textbf{Grid Size (g)} &  \textbf{$ \delta $} & \textbf{Sanity Bound, $ \rho = \delta \cdot \abs{\Dset} $} & \textbf{\QR Size (\%)}  & \textbf{Privacy Budget, $\epsilon$} & \textbf{$\epsilon_1$ (\% of $\epsilon$)} \\
			\hline
			\hline
			Storage, 8,938 points & 30, 40, 50, 60, 70, 80 &  0.1 & 893.8 & \textbf{1, 4, 9, 16, 25, 64} & 1 & 20\\
			\hline
			Landmark, 869,976 points & 200, 250, 300, 350, 400 &  0.1 & 869,9.76 & \textbf{1, 4, 9, 16, 25, 64} & 1 & 20\\
			\hline
			Gowalla, 6,442,841 points & 300, 400, 500, 600, 700, 800 &  0.1 & 6,442.841 & \textbf{1, 4, 9, 16, 25, 64} & 1 & 20\\
			\hline
			Storage, 8,938 points & \textbf{40, 60, 80} &  \textbf{0.002, 0.001, 0.1} & \textbf{17.88, 89.38, 893.8} &  1, 4, 9, 16, 25, 64 & 1 & 20\\
			\hline
			Storage, 8,938 points & 30, 40, 50, 60, 70, 80 & 0.1 & 893.8 & 1 & \textbf{0.2,0.4,0.6, 0.8, 1} & 20\\
			\hline
			Storage, 8,938 points & 30, 40, 50, 60, 70, 80 & 0.1 & 893.8 & \textbf{1, 4, 9, 16, 25, 64} & 1 & \textbf{20, 25, 50, 75}\\
			\hline
			
	\end{tabular}}
\end{table*}

\begin{theorem}\label{theorem:exp-privacy}
	Algorithm~\ref{algo:private-grid} preserves
	$(\epsilon_1+\epsilon_2)$-differential privacy.
\end{theorem}

\begin{proof}
Phase~1 of the algorithm corresponds to the exponential mechanism, in that its
release is sampled according to the exponential mechanism's response
distribution, using the score function \eqref{eq:score-avg}. Since the
algorithm uses response-dependent sensitivity \GS{r} as bounded in
Corollary~\ref{corollary:gs_relative_avg} with privacy parameter $\epsilon_1$,
it preserves $\epsilon_1$-differential privacy by Theorem~\ref{theorem:exp}.
Phase~2 uses the resulting sanitized \gstar which expends no further privacy
budget, but runs the Laplace mechanism with sensitivity 1 (global sensitivity
for histogram release) with privacy parameter $\epsilon_2$. By
Theorem~\ref{theorem:lap} the second phase therefore preserves
$\epsilon_2$-differential privacy. Finally by sequential composition
Lemma~\ref{lem:seq}, the algorithm in total preserves differential privacy
at level $\epsilon_1+\epsilon_2$.
\end{proof}

Our utility guarantee follows from our careful choice of score function,
as itself a bound on algorithm error, combined with utility of the
exponential mechanism~\cite{McSherry07,privacybook14}.

\begin{theorem}\label{theorem:exp-utility}
	Let $\Gset_{OPT} = \{r \in \Gset: s(D, r) = OPT_s(D)\}$ be the set
	of truly optimising grid sizes---\ie each achieves the maximum
	score of all \Gset grid sizes $OPT_s(D) = \max_{r\in\Gset}s(D,r)$; and
	let $\gstar$ be the output of Algorithm~\ref{algo:private-grid}.
	Then for all $\tau>0$
	\begin{eqnarray*}
		\Pr\left[s(\gstar) \le OPT_s(D) - \frac{2\Delta}{\epsilon_1}\left(log\left(\frac{\abs{\Gset}}{\abs{\Gset_{OPT}}}\right) + \tau \right) \right] &\le& e^{-\tau}\enspace,
	\end{eqnarray*}
	where $\Delta=\max_{r\in\Gset}\GS{r}$, each as defined in
	Corollary~\ref{corollary:gs_relative_avg}.
\end{theorem}

With high probability the selected
$\gstar \in \Gset$ has a score close to $OPT_s(D)$ by more than
an additive factor of
$O((\Delta/\epsilon_1) \log\abs{\Gset})$~\cite{McSherry07,privacybook14}
\ie the error has only logarithmic dependence on $\abs\Gset$. Notably
the bound depends on $\epsilon_2$ (in addition to $\epsilon_1$), through $\Delta$.

\subsection{Discussion of Sensitivity Bound}\label{sec:sensitivity-disc}
Conventionally the exponential mechanism is used with a global
bound on score/quality function sensitivity $\Delta$, so as to be
independent of response. 
Following this approach yields two alternative, potentially more conveniently implemented, sensitivity bounds of%they could use the bound in approach~\eqref{eq:gs_rel_app1} or better yet approach~\eqref{eq:gs_rel_app2}:%then the sensitivity of the response can also be bounded as follows

\begin{eqnarray}
	&& \frac{1}{\delta(\delta\left|\Dset\right|+1)} + \frac{\frac{\lambda}{\left|\QRset\right|}\sum_{t \in \QRset} \max_{r} \norm{\Alpha^{r,t}}_{1}}{\delta\left|\Dset\right|(\delta\left|\Dset\right|+1)} + \frac{1}{\delta\left|\Dset\right| + \delta} \label{eq:gs_rel_app2} \\
	&\leq& \frac{1}{\delta(\delta\left|\Dset\right|+1)} + \frac{\lambda \max \{g^2\}}{\delta\left|\Dset\right|(\delta\left|\Dset\right|+1)} + \frac{1}{\delta\left|\Dset\right| + \delta}\enspace, \label{eq:gs_rel_app1}
\end{eqnarray}
where the first bound has removed dependence on grid size by
simply maximising over grid size in the single-query sensitivity 
bound, then averaging. The second sensitivity bound follows from
the observation that the $\|\Alpha^{r,t}\|_1$ terms each quantify
the effective number of cells overlapped by the \QR $t$, which
cannot be any larger than the total number of cells in the histogram.
This in turn is maximised by the grid size with largest number of 
cells. 
\ifdefined\ARXIV
Figure~\ref{fig:maxAlpha} shows that maximising over the grid sizes per \QR, will always yield the largest grid size.

\begin{figure}[b!]
	\centering
	\subfloat[$2 \times 2$ grid \label{fig:2x2grid}]{
		\includegraphics[width = .3\columnwidth]{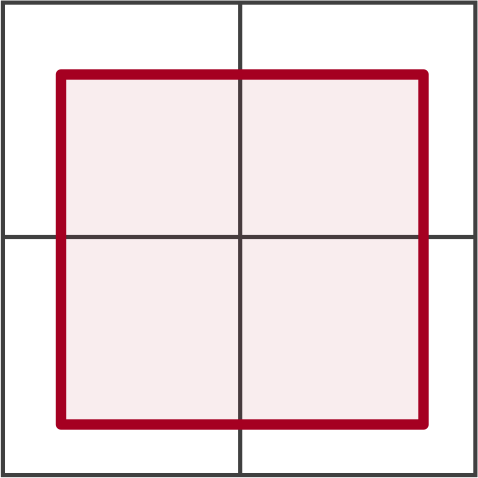}}
	\hspace{.2cm}
	\subfloat[$4 \times 4$ grid \label{fig:4x4grid}]{
		\includegraphics[width = .3\columnwidth]{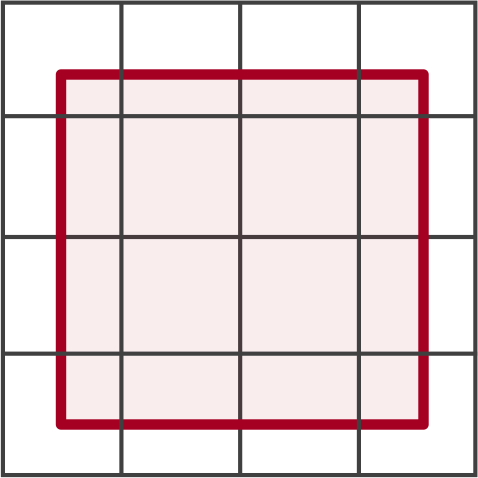}}
	\caption{Maximising over grid sizes, will be always the largest one.} 
	\label{fig:maxAlpha}
\end{figure}
\fi

Both of these alternative approaches would be natural to use with
the exponential mechanism, as response-independent global 
sensitivities. However, they are both upper-bounds on our
response-dependent sensitivity and as such can lead to lower
utility. We demonstrate this effect experimentally in 
Section~\ref{sec:sensitivity-results}.

%% file: subfiles/experiment.tex
\begin{figure*}[t!]
	\centering
	\subfloat[Storage dataset. \label{fig:Storage_QR_error_log}]{
		\includegraphics[scale = .21]{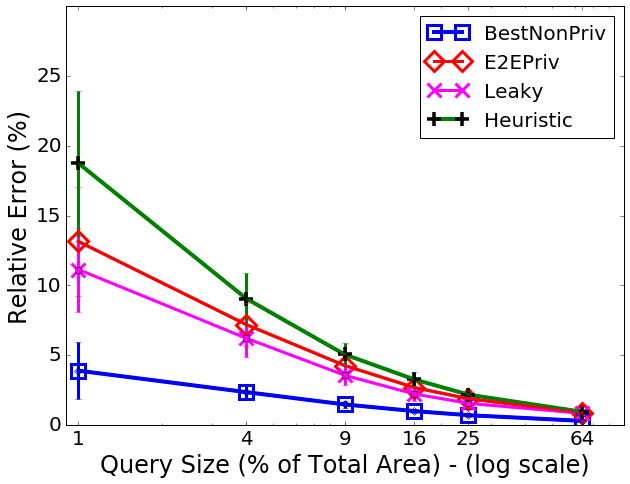}}
	\hfill
	\subfloat[Landmark dataset. \label{fig:Landamrk_QR_error_log}]{
		\includegraphics[scale = .21]{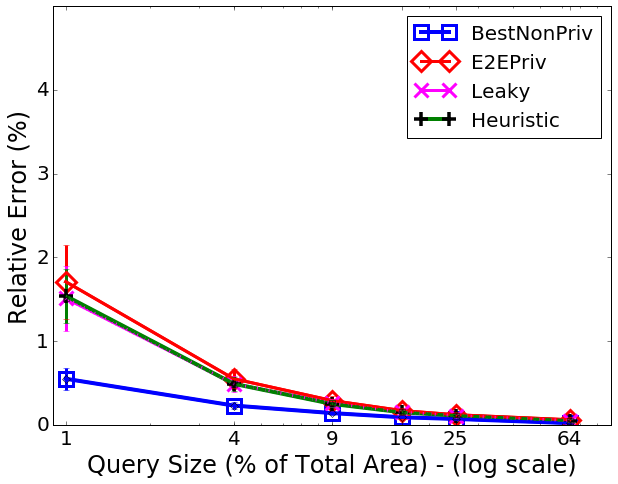}}
	\hfill
	\subfloat[Gowalla Check-ins dataset.\label{fig:Gowalla_QR_error_log}]{
		\includegraphics[scale = .21]{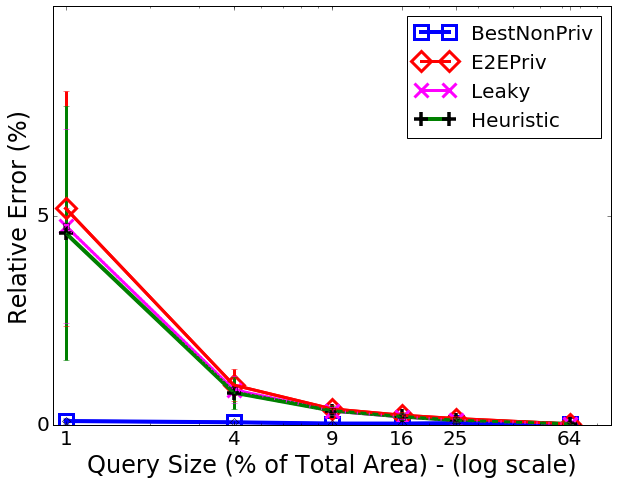}}
	\caption{Effect of QR size (\% of total area) in performance; computing median relative error per query size.}
	\label{fig:datasets_qr_error}
\end{figure*}

\section{Experimental Study}\label{sec:experiment}
We now describe our comprehensive experimental study.

\subsection{Baselines}

We employ three baselines mechanisms in our comprehensive evaluation.
Compared to our truly end-to-end private approach, 
these approaches are either partially private or not differentially
private on new datasets.\\[-0.5em]

\myparagraph{\heuristic}~\cite{QardajiYL13} computes grid size via
Equation~\eqref{eq:heuristic}, as described in
Sections~\ref{sec:related_work} and~\ref{sec:heuristic}. The authors
select $c=10$ based on tuning to the datasets used here. We expect privacy
only when $c$ is not tuned, and as argued, it is not adaptive to the
underlying data, but is based on (partly) principled derivation.\\[-0.5em]

\myparagraph{\leaky} is a semi-private approach, tuning grid size by
adding noise to the original histogram counts (privately) but then
comparing different grid sizes on sensitive data non-privately. 
The entire privacy budget is allocated to the histogram release, none
to (non-private) tuning.\\[-0.5em]

\myparagraph{\best} non-privately releases
the unpertrubed histogram. Without considering noise, tuning optimizes
aggregate error alone, and so always chooses the largest grid size.\\[-0.5em]

\subsection{Datasets}

We run experiments on three datasets---Storage, Landmark, Gowalla
Check-ins---ranging in size, uniformity and sparsity as visualised in
Figure~\ref{fig:datasets}. These datasets were used in~\cite{QardajiYL13}
to evaluate and in fact tune \heuristic (finding $c=10$). In this way, we deliver
\heuristic a significant advantage, providing a fair and comprehensive
comparison between our mechanism and the baselines. 

\ifdefined\ARXIV
Two datasets are in the USA. The first dataset, Storage~\ref{fig:storage}, consists of US storage facility locations composed of national chain storage facilities in addition to locally owned and operated facilities. This is a small dataset of 8,938 points. Geographical coordinates range over (-125.5, -65.5) and (25.0, 50.0) for longitude and latitude respectively. Geographical distances are 60 for Lon axis and 25 for Lat axis. The distance in metres for the x axis is $\approx$6000$Km$, and for y axis is $\approx$2800$km$. The second dataset, Landmark~Figure~\ref{fig:landmark}, is a large dataset of 869,976 points. This is a dataset of locations of landmarks in the 48 US continental states. The listed
landmarks range from schools and post offices to shopping centres, correctional facilities, and train stations from the 2010 Census TIGER point landmarks. As indicated in~\cite{QardajiYL13}, this dataset appears to match the population distribution in the USA. In terms of domain specification, size, longitude and latitude ranges, the dataset is identical to Storage. 

The third and final dataset is the check-in dataset obtained from the Gowalla location-based social network, where users share their locations by checking in. This dataset has the time and location information of check-ins made by users over the period of February 2009--October 2010. For the purpose of this experiment only the location information has been used. This dataset consists of 6,442,841 points, making it a large-sized dataset spanning the entire world map, Figure~\ref{fig:gowalla}. The range for Longitude (x-axis) and Latitude (y-axis) are (80.0, -60.0) and (180.0, -160.0) respectively. Lon axis distance in the geographical system is 340 and Lat is 140, where in the metric system these correspond to $\approx$8,000km and $\approx$16,000km respectively.
\fi

\subsection{Parameter Settings}\label{sec:experiment-param}

Table~\ref{table:experiment} summarises settings, with bolded parameters varying.  
The initial values for the experiment's parameters are as follows. The suggested grid size for Storage by \heuristic Equation~\eqref{eq:heuristic} is $29.88$, which we round to $30$. For Landmark and Gowalla datasets \heuristic selects 295 and 803 respectively. \QR sizes given as input to \leaky and \EE approaches for tuning phases are in the range of $\{.1, .2, .3, .4, .5, .8\}$, which indicates the percentage of domain width and height, \eg .3 means $9\%$ of the total area. We have used the same QR sizes but with different random positions to evaluate all the techniques. $\delta$ to be used for the sanity bound, $\rho$, in bounding relative error during \EE tuning is set to 0.1, 0.01, and 0.001 for Storage, Landmark and Gowalla dataset, respectively (relating to the dataset size: for larger datasets we use smaller $\delta$). $\epsilon$ privacy was initially set to $1$. In terms of allocating privacy budget to our approach \EE, the initial setting was $\epsilon_1=20\%$ and $\epsilon_1=80\%$, which we later vary in Section~\ref{sec:budget}.
Although in the literature~\cite{Cormode12PSD,QardajiYL13} few specific QR sizes are explored, we vary QR's over the entire range of the map area. 
In experiments requiring a fixed QR, we choose the smallest (most challenging) QR of 1\% of total area.

\subsection{Evaluation Metrics}

In our evaluations we use the standard relative error without the
sanity bound---we have no need to control sensitivity (as within our
mechanism) and errors are more interpretable. Similar results are
observed when the sanity bound is introduced. 
The most accurate
method will always be \best as it is non-private, experiences zero
perturbation error and tunes optimally. 
%We hold other effective parameters like $\epsilon$ privacy budget,
%chosen $\delta$ (in \EE tuning) fixed, while varying other parameters,
%to evaluate each technique. For example, varying \QR sizes to observe
%impact on utility. 
Each experiment is repeated 100 times and per \QR size we allocate
100 random positions as our set of query regions. 

\subsection{Effect of Various Query Regions}\label{sec:qr_size}

In this section the median relative error is computed for varying \QR size, to evaluate \EE compared to the baselines. 

Consider first \heuristic, and observe that it can perform well on its experimentally-tuned datasets, Figure~\ref{fig:datasets_qr_error}---recall that it was on these datasets that its $c$ parameter was non-privately tuned. In the results, our approach, \EE, despite being fully differentially private is competitive with \heuristic, and sometimes superior. For Storage Figure~\ref{fig:Storage_QR_error_log}, {\EE}'s error for smallest QR is 13\% while \heuristic only achieves 19\%. This dataset has been chosen by the authors in~\cite{QardajiYL13} to show that their guideline holds for both large and small datasets. However as depicted in Figure~\ref{fig:grid_error} the chosen grid size is not optimal and \EE can outperform the result due to its data-dependence. For Landmark Figure~\ref{fig:Landamrk_QR_error_log}, the error for the smallest QR is less than 2\% and for Gowalla 5\% (Figure~\ref{fig:Gowalla_QR_error_log}). Computed errors for smaller query regions are generally higher, due to the fact that errors for larger queries cancel out. 

% This paragraph may not completely belong to this section, but since we have only discussed Leaky in this section, I have kept it here. This could potentially go to conclucion.
As expected \leaky is always superior to \heuristic and motivates the necessity of having a private tuning technique. It demonstrates that data-dependent tuning improves on grid size selection significantly. However, previous approaches have been non-private. Furthermore, \heuristic is not data dependent and so offers no guarantee it will work. In fact, where it has worked the best, it has been tuned on the data, and not simultaneously private. The existing open problem has been for a mechanism somewhere in between, that is private but data dependent. %It may perform as well as \leaky since it uses part of privacy budget in tuning phase where \leaky does the tuning non-privately. \leaky is more reliable than \heuristic in terms of utility.  Indeed, fixed scheme (\heuristic) is not acceptable due to the variety in datasets in terms of distribution of the points, size of selected map (domain size), sparsity or density of the points.
For the remainder of our experiments, we focus only on the Storage dataset.

%\begin{figure}[t!]
%	\begin{minipage}[t]{1\columnwidth}
%		\centering
%		\subfloat[\small QR Size (\% of Total Area).\label{fig:Storage_QR_error}]{
%			\includegraphics[scale = .19]{figs/Storage_QR_error}}
%		\subfloat[\small QR Size (\% of Total Area) - log scale.\label{fig:Storage_QR_error_log}]{
%			\includegraphics[scale = .19]{figs/Storage_QR_error_log}}
%		\caption{Median relative error per query size for Storage dataset.}
%		\label{fig:storage_qr_error}
%	\end{minipage}
%	\begin{minipage}[t]{1\columnwidth}
%		\centering
%		\subfloat[\small QR Size (\% of Total Area).\label{fig:Landmark_QR_error}]{
%	\includegraphics[scale = .19]{figs/Landmark_QR_error}}
%	\subfloat[\small QR Size (\% of Total Area) - log scale.\label{fig:Landamrk_QR_error_log}]{
%	\includegraphics[scale = .19]{figs/Landmark_QR_error_log}}
%		\caption{Median relative error per query size for Landmark dataset.}
%		\label{fig:landmark_qr_error}
%	\end{minipage}
%	\begin{minipage}[t]{1\columnwidth}
%		\centering
%		\subfloat[\small QR Size (\% of Total Area).\label{fig:Gowalla_QR_error}]{
%		\includegraphics[scale = .19]{figs/Gowalla_QR_error}}
%		\subfloat[\small QR Size (\% of Total Area) - log scale.\label{fig:Gowalla_QR_error_log}]{
%		\includegraphics[scale = .19]{figs/Gowalla_QR_error_log}}
%		\caption{Median relative error per query size for Gowalla dataset..}
%		\label{fig:gowalla_qr_error}
%	\end{minipage}
%\end{figure} 

\subsection{Effect of Privacy Parameter $\epsilon$}\label{sec:epsilon}
We vary the total $\epsilon$ budget to explore its impact on all considered techniques, while fixing \QR to be 1\% of the total area: \cf Figure~\ref{fig:epsilon}. As expected, by increasing privacy, accuracy decreases. However somewhat surprisingly, \EE outperforms \heuristic even though \heuristic has been non-privately tuned on the dataset.

\begin{figure*}[t!]
	\centering
	\subfloat[\label{fig:epsilon}]{
	\includegraphics[scale = .21]{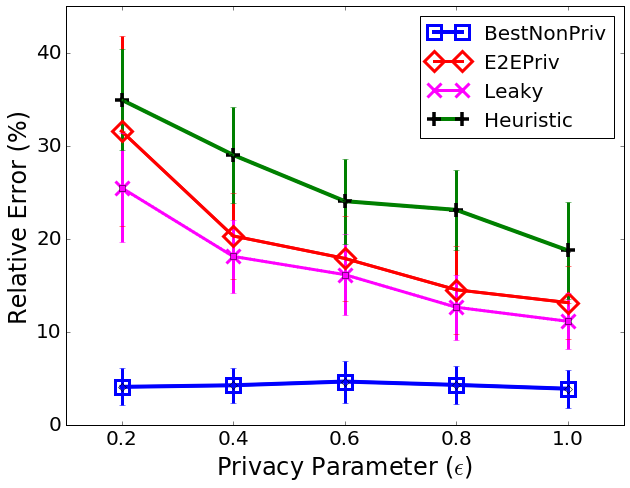}}
	\hfill
	\subfloat[\label{fig:budget}]{
	\includegraphics[scale = .21]{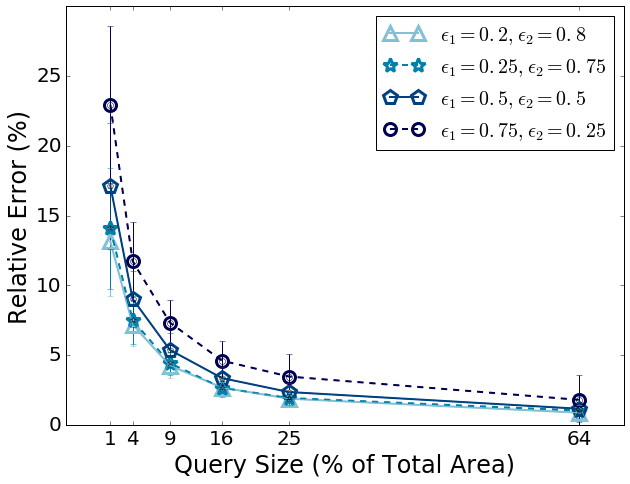}}
	\hfill
	\subfloat[\label{fig:Storage_rho}]{
	\includegraphics[scale = .21]{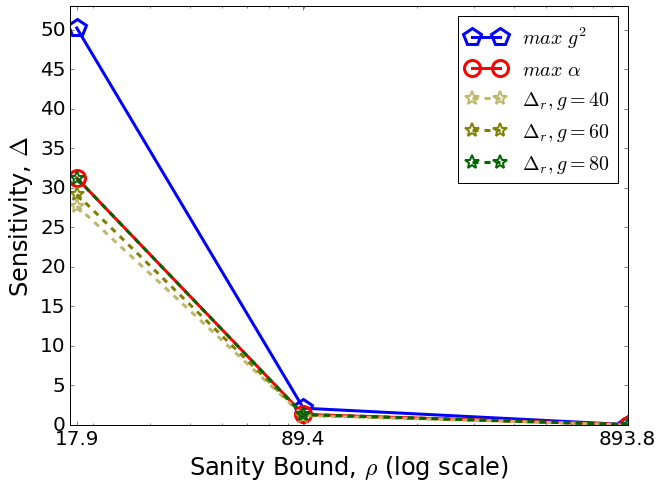}}
	\caption{Effect of (a) varying $\epsilon$ on histogram utility, computed as median relative error; (b) various $\epsilon_1$ and $\epsilon_2$ on histogram utility, computed as median relative error; (c) sanity bound $\rho$ on sensitivity.}
\end{figure*}

\subsection{Effect of Privacy Budget Allocation}\label{sec:budget}
Recall that our approach \EE comprises two phases run sequentially, with total privacy budget split between the two phases. In this section, we demonstrate the effect of different budget allocations to each phase and its impact on the released histogram utility via computing the median relative errors for various test \QR. As shown in Figure~\ref{fig:budget}, for $\epsilon_1 \le 50\%$ the utility of our mechanism remains almost invariant, providing a useful guide for allocating privacy budget.

\subsection{Effect of Sanity Bound $\rho$ on Sensitivity}\label{sec:sensitivity-results}
Figure~\ref{fig:Storage_rho} presents the effect of different $\delta$ parameters, and consequently different $\rho=\delta|\Dset|$, on computed sensitivity bounds by the various approaches derived in Section~\ref{sec:analysis}: our preferred response-dependent bound (Corollary~\ref{corollary:gs_relative_avg}), and the two looser response-independent sensitivities (Equations~\ref{eq:gs_rel_app2} and~\ref{eq:gs_rel_app1}). The results are shown for
grid sizes varying through 40, 60 and 80\%. As shown, the response-dependent $\GS{r}$ does achieve tighter estimates compared to the global alternatives. This difference becomes more significant for reduced sanity bounds, \eg when $\delta=0.002$ yielding $\rho=17.9$. The $\max \alpha$ alternative sees equivalent values to the maximum grid size approach of response dependent sensitivity.

These results confirm our expectation that using more careful response-dependent sensitivity in the exponential mechanism as applied to \EE tuning's Phase~1, can lead to better sensitivity estimates which can in-turn lead to superior utility at no cost to privacy.

%% file: subfiles/conclusion.tex
\section{Concluding Remarks}\label{sec:conclusion}

In this paper we propose a first end-to-end differentially-private mechanism
for releasing parameter-tuned spatial data structures. Our mechanism \EE leverages
a general-purpose concept of tuning via privately-optimising bounds on error:
with the bounds on error derived from utility bounds on the data structure
release mechanism (in this case existing an application of the Laplace mechanism for
releasing histograms); and the private minimisation of these bounds
via the exponential mechanism. Key challenges in
accomplishing our results included the derivation of error bounds and bounding
of these data-dependent error bounds' sensitivity to perturbation.
As a result of our careful analysis, we provide a comprehensive analysis of 
differential privacy and high-probability utility. 

Notably, our bounds on error central to parameter tuning, comprise terms reflecting both
aggregation error due to spatial partitioning and perturbation error due
to post-tuning differential privacy. Our sensitivity calculations
are response-dependent, permitting parameter tuning to achieve superior 
utility at no cost to privacy over coarse, global sensitivity approaches. 

Comprehensive experimental results on datasets of a range of scales,
levels of sparsity and uniformity, establish that our principled tuning-and-release
mechanism achieves competitive utility while preserving end-to-end differential privacy.

In the literature, parameter tuning has been previously accomplished either
non-privately (even tuning differentially-private mechanisms on test data)
or by applying fixed parameter guidelines. We establish that neither style of 
existing approach is sufficient, and that private parameter tuning is achievable,
and efficiently implementable.

%% file: subfiles/appendix.tex
\section{Appendix}\label{sec:appendix}
\subsection{Proof of Theorem~\ref{theorem:gs_relative}}
The theorem's proof proceeds by cases, based on where {\DNset}'s extra
point falls:
\begin{enumerate}[C{a}se 1.]
	\item Outside \QR, outside overlapping cells (Figure~\ref{fig:case1});
	\item Outside \QR, inside overlapping cells (Figure~\ref{fig:case2}); or
	\item Inside \QR, inside overlapping cells (Figure~\ref{fig:case3}).
\end{enumerate}

%The upper bound for Case 3 dominates the the other two as it is the maximum value out of all possible cases.
	We drop superscripts on $\a{i}$ for readability. Our task is to bound
	\begin{eqnarray*}
		&& GS(s_r) = \GS{r} s= \max_{\norm{\Dset - \DNset}_1 \le 1} \left|s(\Dset,r) - s (\DNset, r)\right|\\
		&\le& \left|\frac{1}{\max\{\sum_{i \in \Cset} \d{i}, \rho\}} \left(\left|\sum_{i \in \Cset}\a{i} \c{i} - \sum_{i \in \Cset} \d{i}\right| + \lambda  \norm{\Alpha}_{1}\right) - \right. \\ && \left. \frac{1}{\max\{\sum_{i \in \Cset} d'_{i}, \rho^{\prime}\}} \left(\left|\sum_{i \in \Cset}\a{i} c'_{i} - \sum_{i \in \Cset} d'_{i}\right| + \lambda  \norm{\Alpha}_{1}\right)\right|\\
		& \mbox{where} & \\
		&&  \rho = \delta \times\left|\Dset\right|\mbox{, }\rho  > 1\\
		&&  \rho^{\prime} = \delta \times\left|\DNset\right| = \delta \times(\left|\Dset\right| + 1) = \rho + \delta \\
		%		& \mbox{and} & \\
		%		&&\sum_{i \in \Cset} \d{i} = \begin{cases}
		%				0 \mbox{, }max\{\sum_{i \in \Cset} \d{i}, \rho\} =  \rho\\
		%				\ge 1 \mbox{, }max\{\sum_{i \in \Cset} \d{i}, \rho\}  = \sum_{i \in \Cset} \d{i} \mbox{ or } \rho
		%			\end{cases}\\
		%		&&\sum_{i \in \Cset} d'_{i} = \begin{cases}
		%			\sum_{i \in \Cset} \d{i}\mbox{, } max\{\sum_{i \in \Cset} d'_{i}, \rho + \delta\} = \sum_{i \in \Cset} \d{i} \mbox{ or } \rho + \delta  \\
		%			\sum_{i \in \Cset} \d{i} + 1 \mbox{, } max\{\sum_{i \in \Cset} d'_{i}, \rho + \delta\} = \sum_{i \in \Cset} \d{i} + 1 \mbox{ or } \rho + \delta \\
		%		\end{cases}
	\end{eqnarray*}
	Throughout this proof we are going to have the following expression, which is always positive, so we can drop the absolute:
	\begin{eqnarray*}
		&& \left|\frac{1}{\max\{\sum_{i \in \Cset} \d{i}, \rho\}} - \frac{1}{\max\{\sum_{i \in \Cset} d'_{i}, \rho^{\prime}\}}\right| \\
	 &=& \frac{1}{\max\{\sum_{i \in \Cset} \d{i}, \rho\}} - \frac{1}{\max\{\sum_{i \in \Cset} d'_{i}, \rho^{\prime}\}}
	\end{eqnarray*}
	Each case as enumerated above, has sub-cases based on the value taken by the denominator. 
	
	\myparagraph{Case 1.} When the extra point is outside \QR \& outside overlapping cells:
	\begin{eqnarray*}
		&&\sum_{i \in \Cset} d'_{i} = \sum_{i \in \Cset} \d{i} \mbox{ \&  } c'_{i} = \c{i} \Rightarrow\\
		&& \left|\sum_{i \in \Cset}\a{i} c'_{i} - \sum_{i \in \Cset} d'_{i}\right| + \lambda  \norm{\Alpha}_{1} = \left|\sum_{i \in \Cset}\a{i} \c{i} - \sum_{i \in \Cset} \d{i}\right| + \lambda  \norm{\Alpha}_{1}\\
		&& \left|\sum_{i \in \Cset}\a{i} \c{i} - \sum_{i \in \Cset} \d{i}\right| + \lambda  \norm{\Alpha}_{1} \le \left|\Dset\right| + \lambda  \norm{\Alpha}_{1}\\
		&\GS{r} s \le& \left(\left|\sum_{i \in \Cset}\a{i} \c{i} - \sum_{i \in \Cset} \d{i}\right| + \lambda  \norm{\Alpha}_{1}\right) \cdot \left|\frac{1}{\max\{\sum_{i \in \Cset} \d{i}, \rho\}} - \frac{1}{\max\{\sum_{i \in \Cset} d'_{i}, \rho^{\prime}\}}\right| \\ 
		&\le& \left(\left|\Dset\right| + \lambda  \norm{\Alpha}_{1} \right) \cdot \underbrace{\left(\frac{1}{\rho} - \frac{1}{\rho + \delta}\right)}_{\frac{\delta}{\rho\left(\rho+\delta\right)}=\frac{1}{\left|\Dset\right|\left(\rho+\delta\right)}} \le \frac{1}{\rho+\delta} + \frac{\lambda  \norm{\Alpha}_{1}}{\left|\Dset\right|\left(\rho+\delta\right)}\enspace. \\
	\end{eqnarray*}
	where three possible sub-cases occur:
	%	\begin{align*} \Aboxed{\sum_{i \in \Cset} \d{i} < \rho < \rho + \delta } \end{align*}
	\begin{eqnarray*}
		&\mbox{\textbf{Case 1.1} } & \boxed{\sum_{i \in \Cset} \d{i} < \rho < \rho + \delta }\\
		&& \max\left\{\sum_{i \in \Cset} \d{i}, \rho\right\} = \rho \mbox{ \&  } \max\left\{\sum_{i \in \Cset} d'_{i}, \rho^{\prime}\right\} = \rho+\delta \\
		&& \Rightarrow \frac{1}{\max\{\sum_{i \in \Cset} \d{i}, \rho\}} - \frac{1}{\max\{\sum_{i \in \Cset} d'_{i}, \rho^{\prime}\}} = \frac{1}{\rho} - \frac{1}{\rho + \delta} \\
		&\mbox{\textbf{Case 1.2} } & \boxed{\rho < \rho + \delta < \sum_{i \in \Cset} \d{i}}\\
		&&\max\left\{\sum_{i \in \Cset} \d{i}, \rho\right\} = \sum_{i \in \Cset} \d{i} \mbox{ \&  } \max\left\{\sum_{i \in \Cset} d'_{i}, \rho^{\prime}\right\} = \sum_{i \in \Cset} \d{i}\\
		&& \Rightarrow \frac{1}{\max\{\sum_{i \in \Cset} \d{i}, \rho\}} - \frac{1}{\max\{\sum_{i \in \Cset} d'_{i}, \rho^{\prime}\}} = \frac{1}{\sum_{i \in \Cset} \d{i}} - \frac{1}{\sum_{i \in \Cset} \d{i}} \\
		&=& 0\\
		&\mbox{\textbf{Case 1.3} } &  \boxed{\rho < \sum_{i \in \Cset} \d{i} < \rho + \delta} \\
		&& \underbrace{\max\left\{\sum_{i \in \Cset} \d{i}, \rho\right\} = \sum_{i \in \Cset} \d{i}}_{\sum_{i \in \Cset} \d{i} > \rho \Rightarrow \frac{1}{\sum_{i \in \Cset} \d{i}} < \frac{1}{\rho}} \mbox{ \&  } \underbrace{\max\left\{\sum_{i \in \Cset} d'_{i}, \rho^{\prime}\right\} = \rho+\delta}_{\sum_{i \in \Cset} \d{i} < \rho + \delta \Rightarrow\frac{1}{\sum_{i \in \Cset} \d{i}} > \frac{1}{\rho + \delta}}\\
		&& \Rightarrow \frac{1}{\max\{\sum_{i \in \Cset} \d{i}, \rho\}} - \frac{1}{\max\{\sum_{i \in \Cset} d'_{i}, \rho^{\prime}\}} = \frac{1}{\sum_{i \in \Cset} \d{i}} - \frac{1}{\rho + \delta} \\
		&& \Rightarrow \frac{1}{\max\{\sum_{i \in \Cset} \d{i}, \rho\}} - \frac{1}{\max\{\sum_{i \in \Cset} d'_{i}, \rho^{\prime}\}} \le \frac{1}{\rho} - \frac{1}{\rho + \delta} \\
	\end{eqnarray*}
	
	\myparagraph{Case 2.} When the extra point is outside \QR \& inside overlapping cells, for some $j$:
	\begin{eqnarray*}
		&&\sum_{i \in \Cset} d'_{i} = \sum_{i \in \Cset} \d{i} \mbox{ \&  } c'_{j} = \c{j} + 1 \Rightarrow\\
		&& \left|\sum_{i \in \Cset}\a{i} c'_{i} - \sum_{i \in \Cset} d'_{i}\right| + \lambda  \norm{\Alpha}_{1} = \left|\sum_{i \in \Cset}\a{i} \c{i}  + \a{j} - \sum_{i \in \Cset} \d{i}\right| + \lambda  \norm{\Alpha}_{1}\\
	\end{eqnarray*}
	\begin{eqnarray*}
		&\GS{r} s \le& \left|\frac{1}{\max\{\sum_{i \in \Cset} \d{i}, \rho\}} \left(\left|\sum_{i \in \Cset}\a{i} \c{i} - \sum_{i \in \Cset} \d{i}\right| + \lambda  \norm{\Alpha}_{1}\right) - \right.\\ && \left.\frac{1}{\max\{\sum_{i \in \Cset} d'_{i}, \rho^{\prime}\}} \left(\left|\sum_{i \in \Cset}\a{i} c'_{i} - \sum_{i \in \Cset} d'_{i}\right| + \lambda  \norm{\Alpha}_{1}\right)\right|\\
		&&\mbox{Triangle inequality:}\\
		&\le& \left|\left|\frac{1}{\max\{\sum_{i \in \Cset} \d{i}, \rho\}} \left(\sum_{i \in \Cset}\a{i} \c{i} - \sum_{i \in \Cset} \d{i}\right) \right| - \right.\\
		&& \left.\left|\frac{1}{\max\{\sum_{i \in \Cset} d'_{i}, \rho^{\prime}\}} \left(\sum_{i \in \Cset}\a{i} c'_{i} - \sum_{i \in \Cset} d'_{i}\right)\right|\right| + \\
		&& \left|\frac{1}{\max\{\sum_{i \in \Cset} \d{i}, \rho\}}  \lambda  \norm{\Alpha}_{1}  - \frac{1}{\max\{\sum_{i \in \Cset} d'_{i}, \rho^{\prime}\}}  \lambda  \norm{\Alpha}_{1} \right|\\
		&&\mbox{Reverse trianlge inequality:}\\
		&\le& \left|\frac{1}{\max\{\sum_{i \in \Cset} \d{i}, \rho\}} \left(\sum_{i \in \Cset}\a{i} \c{i} - \sum_{i \in \Cset} \d{i}\right)  - \right.\\
		&& \left. \frac{1}{\max\{\sum_{i \in \Cset} d'_{i}, \rho^{\prime}\}} \left(\sum_{i \in \Cset}\a{i} c'_{i} - \sum_{i \in \Cset} d'_{i} \right)\right| + \\
		&& \lambda  \norm{\Alpha}_{1} \cdot \left|\frac{1}{\max\{\sum_{i \in \Cset} \d{i}, \rho\}} - \frac{1}{\max\{\sum_{i \in \Cset} d'_{i}, \rho^{\prime}\}} \right|\\
		&&\mbox{Rearranging and factoring:} \\
		&\le&\left|\frac{1}{\max\{\sum_{i \in \Cset} \d{i}, \rho\}} \left(\sum_{i \in \Cset}\a{i} \c{i} - \sum_{i \in \Cset} \d{i} \right) - \right.\\ 
		&& \left.\frac{1}{\max\{\sum_{i \in \Cset} d'_{i}, \rho^{\prime}\}} \left( \left(\sum_{i \in \Cset}\a{i} \c{i} - \sum_{i \in \Cset} \d{i} \right) + \a{j}\right) \right| + \\
		&& \lambda  \norm{\Alpha}_{1} \cdot \left|\frac{1}{\max\{\sum_{i \in \Cset} \d{i}, \rho\}} - \frac{1}{\max\{\sum_{i \in \Cset} d'_{i}, \rho^{\prime}\}} \right|\\
		&\le&\left| \left(\sum_{i \in \Cset}\a{i} \c{i} - \sum_{i \in \Cset} \d{i} \right) \left(\frac{1}{\max\{\sum_{i \in \Cset} \d{i}, \rho\}}  - \frac{1}{\max\{\sum_{i \in \Cset} d'_{i}, \rho^{\prime}\}} \right) - \right.\\
		&& \left.\frac{\a{j}} {\max\{\sum_{i \in \Cset} d'_{i}, \rho^{\prime}\}} \right| + \\
		&& \lambda  \norm{\Alpha}_{1} \cdot \left|\frac{1}{\max\{\sum_{i \in \Cset} \d{i}, \rho\}} - \frac{1}{\max\{\sum_{i \in \Cset} d'_{i}, \rho^{\prime}\}} \right|\\
		&&\mbox{Triangle inequality:}\\
		&\le&\left| \sum_{i \in \Cset}\a{i} \c{i} - \sum_{i \in \Cset} \d{i} \right|\cdot \left|\frac{1}{\max\{\sum_{i \in \Cset} \d{i}, \rho\}}  - \frac{1}{\max\{\sum_{i \in \Cset} d'_{i}, \rho^{\prime}\}} \right| + \\
		&& \left|\frac{\a{j}} {\max\{\sum_{i \in \Cset} d'_{i}, \rho^{\prime}\}} \right| + \\
		&& \lambda  \norm{\Alpha}_{1} \cdot \left|\frac{1}{\max\{\sum_{i \in \Cset} \d{i}, \rho\}} - \frac{1}{\max\{\sum_{i \in \Cset} d'_{i}, \rho^{\prime}\}} \right|\\
		&\le&\left( \left|\sum_{i \in \Cset}\a{i} \c{i} - \sum_{i \in \Cset} \d{i} \right|+ \lambda  \norm{\Alpha}_{1}\right) \cdot \\
		&& \left|\frac{1}{\max\{\sum_{i \in \Cset} \d{i}, \rho\}}  - \frac{1}{\max\{\sum_{i \in \Cset} d'_{i}, \rho^{\prime}\}} \right| + \\ 
		&& \left|\frac{\a{i}}{\max\{\sum_{i \in \Cset} d'_{i}, \rho^{\prime}\}} \right|\enspace. \\
		&\le& \left(\left|\Dset\right| + \lambda  \norm{\Alpha}_{1} \right) \cdot \underbrace{\left(\frac{1}{\rho} - \frac{1}{\rho + \delta}\right)}_{\frac{\delta}{\rho\left(\rho+\delta\right)}} + \frac{1}{\rho+\delta} \\
		&\le& \frac{1}{\rho+\delta} + \frac{\lambda  \norm{\Alpha}_{1}}{\left|\Dset\right|\left(\rho+\delta\right)} + \frac{1}{\rho + \delta}\enspace. \\
	\end{eqnarray*} 
	where cases 1.1, 1.2 and 1.3 apply here as well, and result follows.
	\myparagraph{Case 3.} When the extra point is inside both \QR \& overlapping cells:
	\begin{eqnarray*}
		&&\sum_{i \in \Cset} d'_{i} = \sum_{i \in \Cset} \d{i} + 1 \mbox{ \&  } c'_{j} = \c{j} + 1\\
		&&\mbox{From case 2 result, we have the following:}\\
		&\GS{r} s \le& \Big|\frac{1}{\max\{\sum_{i \in \Cset} \d{i}, \rho\}} \left(\sum_{i \in \Cset}\a{i} \c{i} - \sum_{i \in \Cset} \d{i}\right)  - \\
		&& \frac{1}{\max\{\sum_{i \in \Cset} d'_{i}, \rho^{\prime}\}} \left(\sum_{i \in \Cset}\a{i} c'_{i} - \sum_{i \in \Cset} d'_{i} \right)\Big| + \\
		&& \lambda  \norm{\Alpha}_{1} \cdot \left|\frac{1}{\max\{\sum_{i \in \Cset} \d{i}, \rho\}} - \frac{1}{\max\{\sum_{i \in \Cset} d'_{i}, \rho^{\prime}\}} \right|\\
		&\le&\left|\frac{1}{\max\{\sum_{i \in \Cset} \d{i}, \rho\}} \left(\sum_{i \in \Cset}\a{i} \c{i} - \sum_{i \in \Cset} \d{i} \right) - \right.\\ 
		&& \left.\frac{1}{\max\{\sum_{i \in \Cset} d'_{i}, \rho^{\prime}\}} \left( \left(\sum_{i \in \Cset}\a{i} \c{i} - \sum_{i \in \Cset} \d{i} \right) + \left(\a{j} - 1\right)\right) \right| + \\
		&& \lambda  \norm{\Alpha}_{1} \cdot \left|\frac{1}{\max\{\sum_{i \in \Cset} \d{i}, \rho\}} - \frac{1}{\max\{\sum_{i \in \Cset} d'_{i}, \rho^{\prime}\}} \right|\\
		%	&\le&\Big| \big(\sum_{i \in \Cset}\a{i} \c{i} - \sum_{i \in \Cset} \d{i} \big) \Big(\frac{1}{max\{\sum_{i \in \Cset} \d{i}, \rho\}}  - \frac{1}{max\{\sum_{i \in \Cset} d'_{i}, \rho^{\prime}\}} \Big) - \\
		%	&& \frac{\a{j} - 1} {max\{\sum_{i \in \Cset} d'_{i}, \rho^{\prime}\}} \Big| + \\
		%	&& \lambda  \norm{\Alpha}_{1} \cdot \Big|\frac{1}{max\{\sum_{i \in \Cset} \d{i}, \rho\}} - \frac{1}{max\{\sum_{i \in \Cset} d'_{i}, \rho^{\prime}\}} \Big|\\
		%	&\le&\Big| \sum_{i \in \Cset}\a{i} \c{i} - \sum_{i \in \Cset} \d{i} \Big|\cdot \Big|\frac{1}{max\{\sum_{i \in \Cset} \d{i}, \rho\}}  - \frac{1}{max\{\sum_{i \in \Cset} d'_{i}, \rho^{\prime}\}} \Big| + \\
		%	&& \Big|\frac{\a{i} - 1} {max\{\sum_{i \in \Cset} d'_{i}, \rho^{\prime}\}} \Big| + \lambda  \norm{\Alpha}_{1} \cdot \Big|\frac{1}{max\{\sum_{i \in \Cset} \d{i}, \rho\}} - \frac{1}{max\{\sum_{i \in \Cset} d'_{i}, \rho^{\prime}\}} \Big|\\
		&\le&\left( \left|\sum_{i \in \Cset}\a{i} \c{i} - \sum_{i \in \Cset} \d{i} \right|+ \lambda  \norm{\Alpha}_{1}\right) \cdot \\
		&& \left|\frac{1}{\max\{\sum_{i \in \Cset} \d{i}, \rho\}}  - \frac{1}{\max\{\sum_{i \in \Cset} d'_{i}, \rho^{\prime}\}} \right| + \\ 
		&& \left|\frac{1 - \a{j}}{\max\{\sum_{i \in \Cset} d'_{i}, \rho^{\prime}\}} \right|\\
		&\le& \left(\left|\Dset\right| + \lambda  \norm{\Alpha}_{1} \right) \cdot \underbrace{\left(\frac{1}{\rho} - \frac{1}{\rho + 1}\right)}_{\frac{1}{\rho(\rho+1)}} + \frac{1}{\rho+\delta} \\
		&\le& \frac{1}{\delta(\rho+1)} + \frac{\lambda  \norm{\Alpha}_{1}}{\rho(\rho+1)} + \frac{1}{\rho + \delta}\enspace. 
	\end{eqnarray*}
	where again we have sub cases on the denominator.
	\begin{eqnarray*}
		&\mbox{\textbf{Case 3.1} } & \boxed{\sum_{i \in \Cset} \d{i} < \sum_{i \in \Cset} \d{i} + 1  < \rho < \rho + \delta }\\
		&& \max\left\{\sum_{i \in \Cset} \d{i}, \rho\right\} = \rho \mbox{ \&  } \max\left\{\sum_{i \in \Cset} d'_{i}, \rho^{\prime}\right\} = \rho+\delta\\
		&& \Rightarrow \frac{1}{\max\{\sum_{i \in \Cset} \d{i}, \rho\}} - \frac{1}{\max\{\sum_{i \in \Cset} d'_{i}, \rho^{\prime}\}} = \frac{1}{\rho} - \frac{1}{\rho + \delta} \\
		&\mbox{\textbf{Case 3.2} } & \boxed{\rho < \sum_{i \in \Cset} \d{i} < \rho + \delta < \sum_{i \in \Cset} \d{i} + 1  }\\
		&&\underbrace{\max\left\{\sum_{i \in \Cset} \d{i}, \rho\right\} = \sum_{i \in \Cset} \d{i}}_{\sum_{i \in \Cset} \d{i} > \rho \Rightarrow \frac{1}{\sum_{i \in \Cset} \d{i}} < \frac{1}{\rho} } \mbox{ \&  } \underbrace{\max\left\{\sum_{i \in \Cset} d'_{i}, \rho^{\prime}\right\} = \sum_{i \in \Cset} \d{i} + 1}_{\sum_{i \in \Cset} \d{i} + 1 > \rho + \delta \Rightarrow \frac{1}{\sum_{i \in \Cset} \d{i} + 1} < \frac{1}{\rho + \delta}}\\ 
		&& \Rightarrow \frac{1}{\max\{\sum_{i \in \Cset} \d{i}, \rho\}} - \frac{1}{\max\{\sum_{i \in \Cset} d'_{i}, \rho^{\prime}\}} = \frac{1}{\sum_{i \in \Cset} \d{i}} - \frac{1}{\sum_{i \in \Cset} \d{i} + 1}\\
		&& \Rightarrow \frac{1}{\max\{\sum_{i \in \Cset} \d{i}, \rho\}} - \frac{1}{\max\{\sum_{i \in \Cset} d'_{i}, \rho^{\prime}\}} \le \frac{1}{\rho} - \frac{1}{\rho + \delta}  \\
		&\mbox{\textbf{Case 3.3} } & \boxed{\rho < \sum_{i \in \Cset} \d{i} < \sum_{i \in \Cset} \d{i} + 1 < \rho + \delta }\\
		&&\underbrace{\max\left\{\sum_{i \in \Cset} \d{i}, \rho\right\} = \sum_{i \in \Cset} \d{i}}_{\sum_{i \in \Cset} \d{i} > \rho \Rightarrow \frac{1}{\sum_{i \in \Cset} \d{i}} < \frac{1}{\rho} } \mbox{ \&  } \underbrace{\max\left\{\sum_{i \in \Cset} d'_{i}, \rho^{\prime}\right\} = \rho + \delta}_{\sum_{i \in \Cset} \d{i} + 1 < \rho + \delta}\\
		&& \Rightarrow \frac{1}{\max\{\sum_{i \in \Cset} \d{i}, \rho\}} - \frac{1}{\max\{\sum_{i \in \Cset} d'_{i}, \rho^{\prime}\}} = \frac{1}{\sum_{i \in \Cset} \d{i}} - \frac{1}{\rho + \delta}\\
		&& \Rightarrow \frac{1}{\max\{\sum_{i \in \Cset} \d{i}, \rho\}} - \frac{1}{\max\{\sum_{i \in \Cset} d'_{i}, \rho^{\prime}\}} \le \frac{1}{\rho} - \frac{1}{\rho + \delta}  \\
		&\mbox{\textbf{Case 3.4} } & \boxed{\sum_{i \in \Cset} \d{i} < \rho < \rho + \delta < \sum_{i \in \Cset} \d{i} + 1 < \rho + 1}\\
		&& \underbrace{\max\left\{\sum_{i \in \Cset} \d{i}, \rho\right\} = \rho}_{\substack{\sum_{i \in \Cset} \d{i} < \rho \Rightarrow \frac{1}{\sum_{i \in \Cset} \d{i}} > \frac{1}{\rho}\\
				\sum_{i \in \Cset} \d{i} + 1 < \rho + 1\Rightarrow \frac{1}{\sum_{i \in \Cset} \d{i} + 1} > \frac{1}{\rho + 1}} } \mbox{ \&  } \underbrace{\max\left\{\sum_{i \in \Cset} d'_{i}, \rho^{\prime}\right\} = \sum_{i \in \Cset} \d{i} + 1}_{\substack{\sum_{i \in \Cset} \d{i} + 1 > \rho + \delta \Rightarrow\frac{1}{\sum_{i \in \Cset} \d{i} + 1} < \frac{1}{\rho + \delta}\\ \rho + 1 > \rho + \delta \Rightarrow \frac{1}{\rho + 1} < \frac{1}{\rho + \delta}}}\\
		&& \Rightarrow \frac{1}{\max\{\sum_{i \in \Cset} \d{i}, \rho\}} - \frac{1}{\max\{\sum_{i \in \Cset} d'_{i}, \rho^{\prime}\}} = \frac{1}{\rho} - \frac{1}{\sum_{i \in \Cset} \d{i} + 1} \\
		&& \Rightarrow \frac{1}{\max\{\sum_{i \in \Cset} \d{i}, \rho\}} - \frac{1}{\max\{\sum_{i \in \Cset} d'_{i}, \rho^{\prime}\}} \le \frac{1}{\rho} - \frac{1}{\rho + 1} \\	
	\end{eqnarray*}
	The final upper bound for the response-dependent global sensitivity becomes the bound computed for case 3, as it achieves the maximum over all possible cases.